\documentclass[3p,preprint]{elsarticle}

\journal{Information and Computation}

\usepackage{amssymb}

\usepackage[figuresright]{rotating}
\usepackage{enumitem}
\usepackage{hyperref}
\usepackage{multicol}

\newtheorem{theorem}{Theorem}
\newtheorem{lemma}[theorem]{Lemma}
\newtheorem{proposition}[theorem]{Proposition}
\newtheorem{corollary}[theorem]{Corollary}

\newtheorem{example}{Example}

\newproof{proof}{Proof}

\usepackage[utf8]{inputenc}
\usepackage{complexity}
\usepackage[algosection]{algorithm2e}
\usepackage{endnotes}
\usepackage{listings}
\usepackage{mathtools}
\usepackage{verbatim}
\usepackage{mathrsfs}
\usepackage{shuffle}
\usepackage{longtable}
\usepackage{enumitem}
\usepackage{etoolbox}
\usepackage{float}
\usepackage{color}

\usepackage{mathtools}

\newclass{\NCM}{NCM}
\newclass{\NPCM}{NPCM}
\newclass{\NCACM}{NCACM}
\newclass{\NPDA}{NPDA}
\newclass{\NFA}{NFA}
\newclass{\DFA}{DFA}
\newclass{\NQCM}{NQCM}
\newclass{\DCM}{DCM}
\newclass{\NRBQA}{NRBQA}
\newclass{\NRBSA}{NRBSA}
\newclass{\NRBTA}{NRBTA}
\newclass{\NRBQCM}{NRBQCM}
\newclass{\NRBSCM}{NRBSCM}
\newclass{\NRBTCM}{NRBTCM}
\newclass{\NSA}{NSA}
\newclass{\pd}{pd}

\newcommand{\PAIR}{{\rm PAIR}}

\newcommand{\LFam}{{\cal L}}
\newcommand{\MFam}{{\cal M}}
\newcommand{\SFam}{{\cal S}}

\newcommand\acc{{\rm Acc}}
\newcommand\coacc{{\rm co\mbox{-}Acc}}
\newcommand\pre{{\rm pre}}
\newcommand\post{{\rm post}}
\newcommand\conf{{\rm conf}}

\begin{document}

\begin{frontmatter}




\title{On Store Languages and Applications\tnoteref{t1}}

\tnotetext[t1]{\textcopyright 2019. This manuscript version is made available under the CC-BY-NC-ND 4.0 license \url{http://creativecommons.org/licenses/by-nc-nd/4.0/}}

\author[label1]{Oscar H. Ibarra\fnref{fn1}}
\address[label1]{Department of Computer Science\\ University of California, Santa Barbara, CA 93106, USA}
\ead[label1]{ibarra@cs.ucsb.edu}
\fntext[fn1]{Supported, in part, by
NSF Grant CCF-1117708 (Oscar H. Ibarra).}

\author[label2]{Ian McQuillan\fnref{fn2}\corref{corr}}
\address[label2]{Department of Computer Science, University of Saskatchewan\\
Saskatoon, SK S7N 5A9, Canada}
\ead[label2]{mcquillan@cs.usask.ca}
\cortext[corr]{Corresponding author}
\fntext[fn2]{Supported, in part, by Natural Sciences and Engineering Research Council of Canada Grant 2016-06172.}

\begin{abstract}
The store language of a machine of some arbitrary type is the set of all store configurations (state plus store contents but not the input) that can appear in an accepting computation.
New algorithms and characterizations of store languages are obtained, such as the result that any 
nondeterministic pushdown automaton augmented
with reversal-bounded counters, where the pushdown can ``flip'' its contents up to a bounded number of times, can be accepted by a machine with only reversal-bounded counters. 
Then, connections are made between store languages and several model checking and reachability problems, such as accepting the set of all predecessor and successor configurations from a given set  of configurations, and determining whether there are at least one, or infinitely many, common configurations between accepting computations of two machines. These are explored for a variety of different machine models often containing multiple parallel data stores. 
Many of the machine models studied can accept the set of predecessor configurations (of a regular set of configurations), the set of successor configurations, and the set of common configurations between two machines, with a machine model that is simpler than itself, with a decidable emptiness, infiniteness, and disjointness property. Store languages are key to showing these properties. 
 \end{abstract}

\begin{keyword}
Automata \sep Store Languages \sep Counter Machines \sep Deletion Operations \sep Reversal-Bounds \sep Determinism \sep Finite Automata
\end{keyword}

\end{frontmatter}

\section{Introduction}
\label{sec:intro}

An existing concept in the area of formal languages is that of the store language of a machine. Essentially, the store language is the set of store configurations (state plus all store contents concatenated together) that can appear in any accepting computation. For example, the store language of a pushdown automaton is the set of all words
of the form $q \gamma$, where from the initial configuration, there is
an input that passes through the configuration where the 
state is $q$ and the stack contents is $\gamma$, which eventually leads to an accepting
configuration.
It is known that the store language of every one-way nondeterministic pushdown automaton ($\NPDA$) is a regular language \cite{GreibachCFStore,CFHandbook}. This was used by Greibach as a key component of an alternate proof \cite{GreibachCFStore} that regular canonical systems produce regular languages \cite{Buchi1990}.

The store languages of other machine models have been recently studied.
For example, the more general model of one-way nondeterministic stack automata --- which are like pushdown automata but have the additional ability to read but not write from the inside of the pushdown stack --- were investigated, and it was found that the store language of every such machine is also a regular language \cite{KutribCIAA2016,StoreLanguages}.
And, in \cite{StoreLanguages}, the store languages of other one-way machine models were shown to be regular as well, including $r$-flip pushdown automata (pushdown automata with the ability to ``flip'' their pushdown up to $r$ times \cite{flipPushdown}), reversal-bounded queue automata (there is a bound on the number of switches between enqueuing and dequeueing), and nondeterministic Turing machines with a one-way read-only input tape, and a reversal-bounded worktape (a bound on the number of changes of directions of the read/write head). The paper \cite{StoreLanguages} also demonstrated some general connections of store languages between two-way and one-way machine models. Furthermore, it was shown that in any one-way machine model (defined properly) with only regular store languages, then the languages accepted by the deterministic machines in this class are closed under right quotient with regular languages. This solved several open problems in the literature and simplified others. It also demonstrates the usefulness of the store language concept. Similarly, store languages were recently used to show that the density property (whether the subwords of a language are equal to the set of all words) is decidable for Turing machines with a one-way read-only input and a reversal-bounded worktape \cite{DenseJALC}.

Multiple parallel and independent data stores can also be combined into one model. 
However, even a machine that combines
together two pushdowns, has the same power as a Turing machine \cite{HU}, and therefore all
non-trivial problems become undecidable \cite{HU}. 
Another store of interest is that of the counter, which stores some non-negative integer that can be increased by one, decreased by one, kept the same, and tested for emptiness. Equivalence to Turing machines even holds for deterministic machines with only two counters \cite{HU}.
However, if the stores are constrained in some way,
then machines can limit their power and certain properties can become decidable. For example, 
a counter is reversal-bounded if there is a bound on the number of changes between non-decreasing and non-increasing. 
Indeed, the class of one-way nondeterministic
finite automata augmented with some number of reversal-bounded counters (known as $\NCM$), is quite general but has decidable emptiness and membership problems, and is closed under intersection \cite{Ibarra1978}. $\NCM$s have been studied and applied in various places,
e.g., in
\cite{IbarraBultanSu,IbarraSu,IbarraBultan,Finkel,Cadilhac,CadilhacFinkel,NetworksPushdowns}. The deterministic version of these machines, $\DCM$, also has decidable containment and equivalence problems \cite{Ibarra1978}. 
Models can also be created by combining together multiple types of stores, such as $\NPCM$, the set of machines defined by augmenting an (unrestricted) pushdown automaton with reversal-bounded counters. This model also has decidable membership and emptiness 
problems \cite{Ibarra1978}, and has been found useful in showing decidability of
verification and reachability problems \cite{DangIbarraBultan,IbarraDangTCS},
in model checking recursive programs with
numeric data types \cite{HagueLin2011},
in synchronisation- and reversal-bounded analysis of multithreaded programs
\cite{HagueLinCAV}, in showing decidability properties of
models of integer-manipulating programs with recursive
parallelism \cite{HagueLinRP}, and in decidability of problems on commutation \cite{eDCM}.

Separately to the study of store languages, certain similar problems have been studied by the model checking and verification community. The reachability problem in finite-state and infinite-state concurrent systems has been extensively studied (given configurations $c_1$ and $c_2$ of a system, is $c_2$ reachable from $c_1$?) Similarly, we recall two operators that have been extensively studied. Given a set of configurations $C$ and machine $M$ of some type, $\pre_M^*(C)$ is the set of configurations that can reach a configuration in $C$, and $\post_M^*(C)$ is the set that can be reached by a configuration in $C$. 
For example, it is known that given a pushdown automaton $M$ and a regular set of configurations $C$, $\pre_M^*(C)$ and $\post_M^*(C)$ are both regular languages \cite{PushdownVerification}. It is also known that for $\NCM$ $M$ and a set of configurations $C \in \NCM$, both $\pre_M^*(C)$ and $\post_M^*(C)$ are in $\DCM$ \cite{IbarraSu}. These operations have also been studied for other machine models, e.g.\ \cite{multipushdownmodel,Seth,MultiStackModelChecking,Bouajjani,Finkel2000}.

In Section \ref{sec:store}, it is shown that the store language of every $\NPCM$ can be accepted by a machine in $\NCM$; ie.\ without the pushdown. 
This is used to show that all store languages of a new general model, with
a $r$-flip
nondeterministic pushdown automaton
augmented by reversal-bounded counters (denoted by $r$-flip $\NPCM$), are in $\NCM$. 
Hence, the flipping pushdown store can be surprisingly eliminated when accepting the store language. 
New and existing results on store languages are summarized in Table \ref{storeresults}.
Next, in Section \ref{sec:reachable}, 
the notion of the store language is applied to reachability and model checking problems. A simple connection is made between the set of all store languages being in some family, and $\pre^*(C)$ and $\post^*(C)$ being in the family. This new connection is used to demonstrate new reachability results involving several machine models where all store languages are known to be regular, such as stack automata,  and nondeterministic Turing machines with a one-way read-only input and a reversal-bounded worktape. For machines $M$ from these models, $\pre_M^*(C)$ and $\post_M^*(C)$ are always regular
from a given regular set of configurations $C$. In particular, this is is quite surprising for stack automata, as these machines are significantly more powerful than pushdown automata (even accepting non-semilinear languages). 
Furthermore, it is shown that for many models augmented with reversal-bounded counters such as $r$-flip $\NPCM$, given sets of configurations $C$ accepted by a machine in $\NCM$, then $\pre_M^*(C)$ and $\post_M^*(C)$ are in $\NCM$. This implies that the primary store such as the $r$-flip pushdown is again not needed. The same is true for other models augmented by reversal-bounded counters. 
All of these reachability results (summarized in Table \ref{verificationresults}) follow in a straightforward fashion from the connection to store languages, and using results on store languages.

In Section \ref{sec:common}, 
the configurations that are in common between accepting computations of two given machines are examined.
Let $M_1$ and $M_2$ be two $\NPCM$s (respectively $r$-flip $\NPCM$s, $\NCM$s)
with the same pushdown alphabet, the same number of counters,
and the same state sets.
Suppose we are interested in knowing whether the computations of $M_1$ and $M_2$ are completely disjoint in the sense that there is no configuration in common and there is therefore no overlap in computation. 
The common store configuration problem is the following:
Given two machines $M_1$ and $M_2$ of the same type,
do they have a common non-initial store
configuration that occurs in an accepting computation?   Similarly, the common 
store configuration infiniteness problem is  the following:
Given two machines $M_1$ and $M_2$ of the same type,
is there an infinite number of common non-initial
store configurations in accepting computations?   
The common configuration problem can therefore be used to determine whether the computations of $M_1$ and $M_2$ are completely disjoint, and the common store configuration infiniteness problem addresses whether there are finitely many overlapping configurations.
This is related to the notion of fault-tolerance or safety, which are important in the area of verification \cite{IbarraSu}. If $M_2$ is used to describe all faulty configurations (its complement being the safe configurations), then the processing of $M_1$ can lead to a problem (i.e.\ a faulty situation) if and only if they have a common configuration.
For any of the machine models studied in this paper where the store languages are all regular or in $\NCM$, there are immediate applications to these problems, as
the common store configuration problem, and the common store
configuration infiniteness problem are decidable. This is decidable for machine models such as $r$-flip $\NPCM$ since their store languages are in $\NCM$. Moreover, an $\NCM$ can be built that accepts exactly those reachable configurations in common between the two machines, i.e.\ it is possible to build a description of exactly the faulty configurations in $M_1$ and to test any desirable properties within it. Such results would not have been
possible without previously demonstrating that the store languages of these
models could be accepted with only the counters, as e.g.\ the context-free
languages are not closed under intersection, but $\NCM$ is closed
under intersection, which is needed to construct the configurations that are in
common.

In Section \ref{sec:pairs}, 
problems are studied involving the following basic reachability problem:
given configurations $c_1$ and $c_2$, is $c_2$ reachable from $c_1$? More generally, how difficult
is it to accept pairs of configurations $c_1, c_2$ where the second is reachable from the first?
We explore differences in accepting this set based on the representation of the pairs of configurations,
such as whether they are input on two separate tapes, or as separate segments of one tape, or whether certain
configurations are reversed or not.

All of the models studied here are now amongst the most general multi-store models known where these model checking problems are decidable. Furthermore, the connections with store languages provides often short and quite simple proofs of these properties without relying on ad hoc techniques.

\section{Preliminaries}
\label{sec:prelims}

Background knowledge from the area of automata and formal languages is assumed \cite{HU}.
An {\em alphabet} $\Sigma$ is a finite set of symbols. The set of all strings (or words)
over $\Sigma$ is denoted by $\Sigma^*$. A {\em language} $L$ over $\Sigma$ is any $L \subseteq \Sigma^*$. The empty word is denoted by $\lambda$. 
A language $L \subseteq \Sigma^*$ is {\em bounded} if there exists words $w_1, \ldots, w_l \in \Sigma^*$ such
that $L \subseteq w_1^* \cdots w_l^*$.
Given a word $w\in \Sigma^*$, $|w|$ is the length of $w$, and $|w|_a$ is the number of 
$a$'s in $w$, for $a \in \Sigma$. For a fixed alphabet $\Sigma = \{a_1, \ldots, a_k\}$,
the Parikh map of $w$, $\psi(w) = (|w|_{a_1}, \ldots, |w|_{a_k})$, and the Parikh map
of a language $L$, $\psi(L) = \{\psi(w) \mid w \in L\}$.

Given $L_1, L_2 \subseteq \Sigma^*$, the left inverse (or left quotient) of $L_2$ by $L_1$, is
$(L_1)^{-1} L_2 = \{y \mid xy \in L_2, x \in L_1\}$.
Given a word $w \in \Sigma^*$, the reverse of $w$,
$w^R$, is the word obtained by reversing the letters of $w$. Given $L$, 
$L^R = \{w^R \mid w \in L\}$.
Although we will not define semilinear sets and languages formally here (see \cite{harrison1978}), an equivalent characterization will be stated that is enough for our purposes. A language $L$ is semilinear if and only if
$\psi(L) = \psi(L')$ for some regular language $L'$ \cite{harrison1978}.

For some machine models considered in this paper, an intuitive description of the model will be given rather than a formal definition. This is done as many of the models are familiar to those in the area, and the detail given is enough to understand how they operate. If the reader desires further details, the formal definitions can be found in our recent paper \cite{StoreLanguages}.

It is common in automata theory to study a one-way deterministic or nondeterministic
finite automaton (denoted by $\DFA$ or $\NFA$ respectively) with one or more
of some type of data stores. For example, a nondeterministic pushdown automaton
($\NPDA$) is an $\NFA$ together with a pushdown stack \cite{HU}. A counter
can be thought of as a pushdown with only a single pushdown letter plus
a bottom-of-counter marker to allow for testing if the counter is zero. 
A machine $M$ with $k$ counters is $l$-reversal-bounded if every counter makes
at most $l$ changes in direction between non-decreasing and non-increasing and vice versa.
Let $\NCM(k)$ be the set of all one-way nondeterministic
machines with $k$ counters that are $l$-reversal-bounded, for some $l$, and let $\DCM(k)$
be those machines that are deterministic. Let 
$\NCM = \bigcup_{k \geq 1} \NCM(k)$, and $\DCM = \bigcup_{k \geq 1} \DCM(k)$.
These machines have been extensively studied in \cite{Ibarra1978}.
One can also study machines with a pushdown, plus some
number of reversal-bounded counters. The set of all one-way nondeterministic
machines of this form is called $\NPCM$ (defined in \cite{Ibarra1978}).

The store language of a machine $M$, denoted by $S(M)$, 
is the set of configurations that can appear in any accepting computation of $M$. 
Each configuration is represented by the concatenation of the state, followed by the concatenation of each store's
contents, making it a language.
The precise
definition of the store language $S(M)$, where $M$ is from some machine model
$\MFam$, depends on the definition of the model. In \cite{StoreLanguages},
the store languages of many different models are defined in a general
fashion by separating the definition of ``store types'' from machines using these types. 
A machine of any type is denoted by a tuple $M = (Q, \Sigma, \Gamma, \delta, q_0, F)$, where
$Q$ is the finite state set, $\Sigma$ is the input alphabet, $\Gamma$ is the store alphabet,
$\delta$ is the finite transition function, $q_0 \in Q$ is the initial state, and $F \subseteq Q$ is the final state set.
The transition function $\delta$ is a function that maps a state, an input letter (or $\lambda$, or the right input end-marker $\lhd$), and a letter
read off of each store, to a set of possible successors, each consisting
of a new state, and some allowable instruction for manipulating each store. 

Some definitions will be given for $\NPCM$ specifically since it is used frequently here.
For $\NPCM$s  with $k$ counters (written as $\NPCM(k)$), the transition function can read the top of the pushdown, and the top of each counter detecting whether each counter is empty
or non-empty, and the allowable instructions can replace the topmost symbol with some word, and each counter can either increase by one, decrease by one, or stay the same, so long as the counters remain reversal-bounded.
The transitions are from $\delta(q,a,X,y_1, \ldots, y_k), q \in Q, a \in \Sigma \cup \{\lambda,\rhd\}, X \in \Gamma$ (read from the pushdown), $y_i \in \{0,1\}, 1 \leq i \leq k$ (applied if either $y_i=0$
and counter $i$ is empty, or if $y_i =1$ and counter $i$ is non-empty), to a set
of tuples of the form $(p,\alpha, z_1, \ldots, z_k), p \in Q, \alpha \in \Gamma^*$ (replacing $X$ on the top of the pushdown), and $z_i \in \{-1,0,+1\}$ (either subtracting, keeping the same, or adding to the counter).

An instantaneous description of an $\NPCM(k)$ $M = (Q, \Sigma, \Gamma, \delta, q_0, F)$ is a tuple
$(q,w, \gamma, i_1, \ldots, i_k)$, where $q$ is the current state, $w \in \Sigma^*\lhd \cup \{\lambda\}$
is the remaining input (followed by the input end-marker which is important for deterministic machines but is not needed for nondeterministic machines \cite{eDCM}), $\gamma \in Z_0 (\Gamma - \{Z_0\})^*$ is the current contents of the pushdown (starting with $Z_0$,
the bottom-of-stack marker which is not allowed to be popped or replaced), and $i_1, \ldots, i_k \in \mathbb{N}_0$ (the non-negative integers) are the values
of the $k$ counters.
A derivation relation, $\vdash_M$, is defined between pairs of successive instantaneous descriptions, extended to zero or more applications,
$\vdash_M^*$, in the usual fashion \cite{StoreLanguages}.
A store configuration of $M$ is any string $q \gamma c_1^{i_1} \cdots c_k^{i_k}$, where
$q \in Q, \gamma \in Z_0 (\Gamma-\{Z_0\})^*$, $i_1, \ldots, i_k \in \mathbb{N}_0$. 
That is, it is
the string obtained from an instantaneous description by concatenating the state and store contents (and
not including the input). The relation $c \Rightarrow_M c'$ is used to indicate that store configuration $c$ can
be transformed into $c'$ by one transition, and $\Rightarrow_M^*$ is the reflexive, transitive closure of $\Rightarrow_M$.
The set of all store configuration strings is denoted by $\conf(M)$.
Since each store configuration is a string, $\conf(M)$ is a regular language.
The language accepted by $M$ is the set
$$L(M) = \{w \mid (q_0, w\lhd, Z_0, 0, \ldots, 0) \vdash_M^* 
(q_f, \lambda, \gamma', i_1', \ldots, i_k'), q_f \in F\},$$ and
the store language of $M$ is the set
$$S(M) = \{ q \gamma c_1^{i_1} \cdots c_k^{i_k} \mid (q_0, w \lhd, Z_0, 0, \ldots, 0) \vdash_M^* (q, w', \gamma, i_1, \ldots, i_k) \vdash_M^* (q_f, \lambda, \gamma', i_1', \ldots, i_k'), q_f \in F\}.$$ 
Alternatively, $S(M) = \{ c \in \conf(M) \mid c_0 \Rightarrow_M^* c \Rightarrow_M^* c', c_0$ is the initial configuration, $c'$ is
a final configuration$\}$.
That is, the store language is the set of all store configurations that can occur during an accepting computation.
Given a machine $M$, and a set of configurations $C \subseteq \conf(M)$, the set of {\em predecessors} of $C$ is the set
$$\pre_M^*(C) = \{ c \mid c \Rightarrow_M^* c', c' \in C\},$$ and the set of {\em successors} of $C$ is the set
$$\post_M^*(C) = \{ c' \mid c \Rightarrow_M^* c', c \in C\}.$$

\begin{example}
Consider the language $L = \{w \$ w^R \mid w \in \{a,b\}^*, |w|_a = |w|_b \}$. An $\NPCM(2)$ can
be built to accept $L$ with
$M = (Q,\Sigma,\Gamma,\delta,q_0,F)$, $F = \{q_f\}$ and transitions as follows:
\begin{eqnarray*}
(q_0,Xa,1,0) & \in  & \delta(q_0,a,X,y,z), \forall X \in \{Z_0,a,b\}, \forall y,z \in \{0,1\},\\
(q_0,Xb,0,1) & \in  & \delta(q_0,b,X,y,z), \forall X \in \{Z_0,a,b\}, \forall y,z \in \{0,1\},\\
(q_1,X,0,0) & \in  & \delta(q_0,\$,X,y,z), \forall X \in \{Z_0,a,b\}, \forall y,z \in \{0,1\},\\
(q_1,\lambda,0,0) & \in  & \delta(q_1,X,X,y,z), \forall X \in \{a,b\}, \forall y,z \in \{0,1\},\\
(q_2,Z_0,0,0) & \in  & \delta(q_1,\rhd, Z_0,y,z), \forall y,z \in \{0,1\},\\
(q_2,Z_0,-1,-1) & \in  & \delta(q_2,\lambda, Z_0,1,1),\\
(q_f,Z_0,0,0) & \in  & \delta(q_2,\lambda, Z_0,0,0).
\end{eqnarray*}
On input $w\$v$, $M$ reads $w$ and pushes it onto the stack while incrementing the first counter to $|w|_a$ and the second counter to $|w|_b$. Then $M$ reads $\$$ and reads $v$ while verifying using the stack that $v = w^R$. On the end-marker $\rhd$, $M$ decrements both counters in parallel and verifies that they hit zero at the same time before accepting. The counters are $1$-reversal-bounded since they never increase after they have been decreased.

For the store language $S(M)$, it consists of all configurations that can appear in an accepting computation. This is
\begin{eqnarray*}
S(M) & = & 
\{ q_0 Z_0 \gamma c_1^i c_2^j \mid \gamma \in \{a,b\}^*, i = |\gamma|_a, j = |\gamma|_b\} \cup 
\{q_1 Z_0 \gamma c_1^i c_2^i \mid \gamma \in \{a,b\}^*, i \geq 0, |\gamma|_a \leq i, |\gamma|_b \leq i \} \cup \\
&&\{q_2 Z_0 c_1^i c_2^i \mid i \geq 0\} \cup \{q_3Z_0\}.
\end{eqnarray*}
Indeed, in an accepting computation, in $q_0$, $M$ could be in any configuration where the counter values match the number of $a$'s in the input and the number of $b$'s in the input, and the stack matches $Z_0$ followed by the input. Note that $i$ does not have to equal $j$ since there is additional letters from $\{a,b\}^*$ that could be read to make the number of $a$'s match the number of $b$'s before $\$$. However, in $q_1$, after $\$$ has been read, $i$ and $j$ must match in any accepting computation, but $\gamma$ can be arbitrary as long as there is at most $i$ $a$'s and $b$'s. In state $q_2$, the stack has been popped, but the contents of both counters must be equal in any accepting computation.
\end{example}

The store language and the set of successors and predecessors can be defined similarly 
for all other models considered here (as per \cite{StoreLanguages}).
The model $r$-flip $\NPDA$ \cite{flipPushdown} can be defined as an $\NPDA$ with an additional instruction called `flip' which allows the machine to reverse the contents of the pushdown above the bottom-of-pushdown marker, and the machine can do this at most $r$ times \cite{StoreLanguages}. This allows non-context-free languages such as $\{w\$w\mid w \in \{a,b\}^*\}$ to be accepted by using a flip instruction after reading $\$$. A nondeterministic queue automaton has an enqueue and dequeue instruction. Although such a machine has the same power as a Turing machine, if a bound is placed on the number of switches between enqueueing and dequeueing, called reversal-bounded, then the power is more limited \cite{Harju}. The model $\NRBQA$ are reversal-bounded queue automata. A stack automata, denoted by $\NSA$, is similar to a pushdown automaton with the additional ability to read from the inside of the stack in a two-way read-only fashion \cite{StackAutomata}. Upon returning to the top of the stack, it can again push and pop. A reversal-bounded stack automaton, $\NRBSA$, has a bound on the number of changes between pushing and popping, but also the number of changes of direction when reading from the inside of the stack \cite{StoreLanguages}.
Here, we also consider a nondeterministic Turing machine to have a one-way read-only input tape and a bi-infinite read/write worktape. If there is a bound on the number of switches in direction on the worktape that it makes from left-to-right or vice versa, then the machine is reversal-bounded. Let $\NRBTA$ be the reversal-bounded Turing machines. Certainly $\NRBTA$ are more general than $\NRBQA$ and $\NRBSA$ in terms of languages accepted.

Each of the models above are also considered by augmenting them with reversal-bounded counters, with notation of the models listed in Table \ref{notation}. For those models with counters, following the notation by $(k)$ (such as $\NPCM(k)$) indicates that there are $k$ counters. This technique is a powerful one, and given a machine model defined properly that only accepts semilinear languages, augmenting them with reversal-bounded counters yields only semilinear languages with a decidable emptiness problem \cite{Harju,CIAA2018}. All of these machines except for $r$-flip $\NPCM$ have been previously considered in the literature \cite{StoreLanguages,Harju}.

\begin{table}
\caption{The one-way nondeterministic machine models considered in this paper are listed below, the notation used with and without counters is provided (`---' means the model is not considered here).}
\begin{center}
\begin{tabular}{l | l | l }
machine model & without counters & with counters\\ \hline
finite automata & $\NFA$ & $\NCM$ \\
pushdown automata & $\NPDA$ & $\NPCM$ \\
$r$-flip pushdown automata & $r$-flip $\NPDA$ & $r$-flip $\NPCM$\\
stack automata & $\NSA$ & --- \\
reversal-bounded stack automata & $\NRBSA$ & $\NRBSCM$\\
reversal-bounded queue automata & $\NRBQA$ & $\NRBQCM$\\
Turing machine with reversal-bounded worktape & $\NRBTA$ & $\NRBTCM$\\
\end{tabular}
\end{center}
\label{notation}
\end{table}%

Given a machine model $\MFam$, the family of languages accepted by machines
in $\MFam$ is denoted by $\LFam(\MFam)$ and the family  of store languages of machines in
$\MFam$ is denoted by $\SFam(\MFam)$. 
Define $\LFam(\REG)$ to be the family of regular languages.
A language family is a {\em trio} if it is closed under $\lambda$-free homomorphism, inverse homomorphism, and
intersection with regular languages \cite{HU}. A family is semilinear if all languages in it are semilinear.

\section{Store Languages of $\NPCM$s and $r$-flip $\NPCM$s}
\label{sec:store}

In \cite{StoreLanguages}, several types of machine models with reversal-bounded main stores plus reversal-bounded counters were studied. It was found that the models $\NRBSCM$, $\NRBQCM$, and $\NRBTCM$ only had store languages in $\LFam(\NCM)$. In this section, it will be shown that the store languages of all $\NPCM$ (the pushdown is unrestricted) and even $r$-flip $\NPCM$ machines are in $\LFam(\NCM)$. 
This is a strong result as it is known that some stack automata have store languages outside $\LFam(\NCM)$ (and certainly queue automata and Turing machines do as well).

The result for $\NPCM$ can be shown
by two approaches, the first using an existing lengthy technique from \cite{IbarraDangTCS} that shows that all ordered pairs of configurations of an $\NPCM$ where the
second configuration is reachable from the first can be accepted by a 2-tape $\NCM$. However, we instead present a direct approach, which is of
independent interest as a technique for studying store languages and verification operations.
The result for $r$-flip $\NPCM$s will then use the result for $\NPCM$s.

First, two definitions are needed. Given an $\NPCM(k)$ $M$, let
$\acc_M(q)$ be the set of all configurations
in state $q$, $q\gamma c_1^{i_1} \cdots c_k^{i_k}$ reachable from the initial configuration, where $\gamma$ is a word over the pushdown alphabet.
Similarly, let $\coacc_M(q)$ be the set of all configurations 
$q\gamma c_1^{i_1} \cdots c_k^{i_k}$
with state $q$ that can eventually reach an accepting configuration.

It will be shown that $\acc_M(q)$ and $\coacc_M(q)$ are in $\LFam(\NCM)$, for all $q$,
and from this, the proof that $S(M)$ is in $\LFam(\NCM)$
will easily follow.

First, a normal form is presented for store languages of $\NPCM$.
For this lemma, a generalized sequential machine (gsm) is used, which is akin to a nondeterministic finite automaton with output \cite{HU}. 
While it is easy to show that every language in $\LFam(\NPCM)$ can be accepted
by a machine in this normal form, it is not possible to show that every store language of machines in $\NPCM$ is the store language of a machine of this form. For example, there are store
languages of $\NPCM$s where multiple states at the beginning of the computation can
exist with a zero for a counter value. But a gsm can adjust for these differences.
\begin{lemma}
\label{NPCMnormalform}
Let $M \in \NPCM$. Then there exists $M' \in \NPCM$
where 
\begin{enumerate}
\item all counters are $1$-reversal-bounded in every
computation, 
\item \label{disjointreversal} the states that are used
before and after each counter reversal are disjoint,
\item from the initial state $q_0$ all counters immediately increase  while keeping $Z_0$ on the pushdown to a new state $q_0'$, 
\item $M'$ only accepts with all counters and pushdown empty (on $Z_0$) in a unique final state $f$,
\item every transition on the pushdown either pops, replaces
the top symbol of the pushdown with a symbol, or pushes one new symbol on the pushdown
(while not changing the symbol beneath),
\end{enumerate}
and there exists a generalized sequential machine (gsm) $g$ such
that $g(S(M')) = S(M)$.
\end{lemma}
\begin{proof}
Let $M$ have $k$ counters that are $l$-reversal-bounded.
Let $r = \lceil l/2 \rceil$.
Then, construct an $r k$ counter machine $M'$. In it, counter $j$ of $M$, for
$1 \leq j \leq k$, is simulated by using counters $(j-1) r+1, \ldots, (j-1)r + r = jr$,
where the first is used until the second counter reversal, then the next one until
the fourth counter reversal, etc. $M'$ simulates $M$, keeping track in the state 
of $M'$ both the simulated state $q$ of $M$, the current counter being used, from $1$ to $r$ for each of the $k$ original counters, and whether the counter reversal has occurred or not (the simulated state
$q$ of $M$ can be uniquely determined and output by the gsm $g$ constructed below). 
When $M'$ switches from one counter to the next after a counter reversal, $M'$
simultaneously subtracts $1$ from the first counter while adding $1$
to the next until the first is empty. Then, when the first counter is empty, this
allows the next counter to continue in the simulation. When $M$ switches to
a final state $q_f$, $M'$ nondeterministically either stays in $q_f$
(allowing the simulation to continue),
or switches to a new state $\bar{q_f}$ from where $M'$ empties all stores before switching
to a new unique final state $f$ of $M'$.

For the pushdown, a transition that replaces the top
of the pushdown $b$ with $cx$, $b,c \in \Gamma, x \in \Gamma^+$,
is simulated by first replacing $b$ with $c$, then pushing
$x$ one symbol at a time (the gsm $g$ constructed below does not output intermediate states used as these are not configurations of $M$).

Furthermore, $M'$ immediately increases every counter to one, and while simulating
$M$, it keeps track of which counters of $M$ have always been zero, and which have
not
(for these counters, the first increase is ignored).

Lastly, create a gsm $g$ that operates as follows on a word of $S(M')$ of the form
$q' w  c_1^{i_1} \cdots  c_{rk}^{i_{rk}}$: Then $g$ outputs
$q w  c_1^{i_1 + \cdots + i_r} \cdots c_{k}^{i_{(k-1)r+1} + \cdots + i_{rk}}$
where $q$ is the state of $M$ determined from $q'$ in the construction above. 
Indeed, the sum of counters $(j-1)r+1, \ldots, (j-1)r+r$ of $M'$ in the simulation is
the same as counter $j$ of $M$, for all $j$, $1 \leq j \leq k$.
However,
$g$ does not output on any state $\bar{q_f}$ of $F$ or any intermediate
states used when simulating the push transitions, or the intermediate
states associated with moving counter contents from one counter to the next. 
For all non-initial
counter values where the state implies that a counter has always been zero,
one fewer $c_i$ is output.
Hence, $g(S(M')) = S(M)$.
\qed
\end{proof}

Two definitions are needed for the next lemma which shows that $\acc_M(q)$ can be 
accepted
by an $\NCM$, for each $q \in Q$.
Let $M  = (Q, \Sigma,\Gamma,\delta,q_0,F)\in \NPCM$ (with bottom-of-stack marker
$Z_0 \in \Gamma$, and $\Gamma_0 = \Gamma - \{Z_0\}$). Then $q \xrightarrow{Q}_t p$ if
transition $t$ switches from state $q$ to $p$.
Let $\alpha = t_1, \ldots, t_n$ be a sequence of transitions. Then
$q \xrightarrow{Q}_{\alpha} p$ if $q \xrightarrow{Q}_{t_1} q_1
\xrightarrow{Q}_{t_2} q_2 \cdots q_{n-1} \xrightarrow{Q}_{t_n} p$,
for some $q_1, \ldots, q_{n-1}$.
Let $x,y \in Z_0 \Gamma_0^* \cup \Gamma_0^*$. Then
$x \xrightarrow{\Gamma}_t y$, $t$ a transition of $M$, if $x$
in a pushdown is changed by $t$ to $y$ (by changing the rightmost symbol of $x$ according to
$y$). Notice that the first letter of $x$ is $Z_0$ if and only if the first letter of $y$ is $Z_0$ since
it cannot be popped or replaced by $M$.
Let $b,c \in \Gamma$, and $\alpha = t_1, \ldots, t_n$
be a sequence of transitions. Define 
$b \xrightarrow{\Gamma}_{\alpha} c$ if $x_0, \ldots, x_n \in \Gamma^*$ are such
that $ b = x_0 \xrightarrow{\Gamma}_{t_1} x_1 
\xrightarrow{\Gamma}_{t_2} \cdots \xrightarrow{\Gamma}_{t_n} x_n = c$.
Again, $b = Z_0$ if and only if $c = Z_0$. Note that when $b$ and $c$ are not $Z_0$, $b \xrightarrow{\Gamma}_{\alpha}c$ if, when starting a pushdown with only symbol
$b$ in the pushdown (no bottom-of-stack marker), and applying the pushdown instructions
in $\alpha$, then the pushdown ends with only $c$. This implies
that the pushdown never empties, but could be larger in intermediate
configurations, but then it needs to eventually end with exactly the symbol
$c$ on the pushdown.
\begin{figure}[htbp]
\begin{center}
\includegraphics[width=11cm]{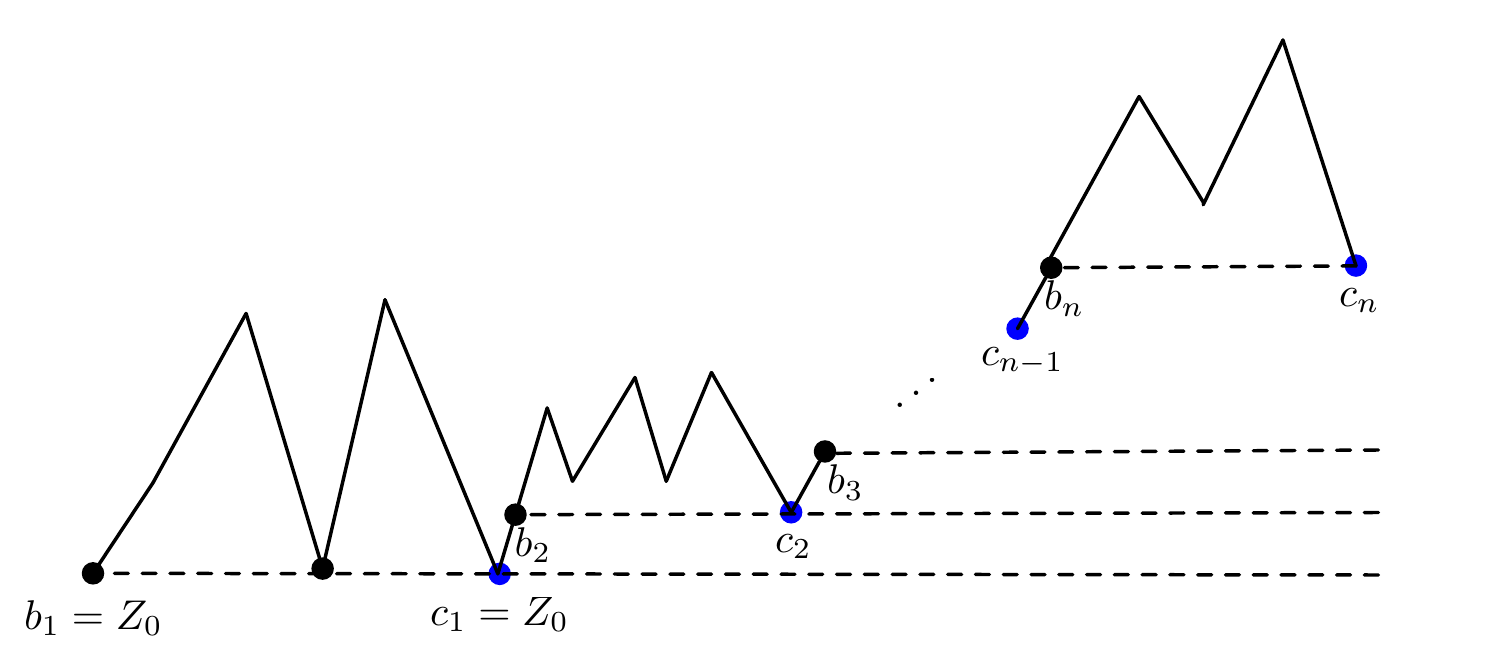}
\caption{The figure is a depiction of a pushdown changing from an initial configuration to a configuration with $Z_0 c_2 \cdots c_n$ on the pushdown. The blue dots indicate that this is the last configuration of this height, and the dotted lines indicate that the pushdown will not go below this line again during this part of the computation. }
\label{fig1}
\end{center}
\end{figure}

Intuitively, the proof operates as follows. Consider a computation of
$M$ that is reachable from the initial configuration, that ends with pushdown contents $\gamma$ of length $n$. Then, for each $i$,
from $1$ to $n$, there is some last configuration where the computation hits
a pushdown of length $i$. For example, there is some last time when the 
computation hits a pushdown of length $1$. So, from the initial configuration to
that configuration, the pushdown can go up in the middle, but it eventually
returns to a pushdown of length $1$. But after that, the pushdown immediately
gets increased to length $2$, and never again returns to a pushdown of length $1$.
Similarly, from this configuration of size $2$, eventually the computation hits
a final time where the pushdown is of size $2$, and between these two configurations
the pushdown can get larger than length $2$, but never goes below length $2$. This is similar all the
way to $n$, and is illustrated in Figure \ref{fig1}. 
Between these pairs of configurations where the stack starts
and ends with a stack of the same size,
an intermediate $\NPDA$ can be built that accepts sequences of symbols associated with the transitions
of $M$ (this does not do any counting with the counters) that can cause the pushdown
to start with one symbol, possibly go up and come back down to the same size, and return to another stack
symbol. This $\NPDA$ is not simulated by the final $\NCM$ directly though, as
$\NCM$ machines do not have a pushdown. Instead, a property of context-free languages is exploited; it is known that all context-free languages are semilinear, and therefore have the same Parikh map as a regular language \cite{harrison1978} (this can effectively construct a $\DFA$).
Therefore, it is possible to make a $\DFA$ that accepts a language with the
same number of each transition symbol as a word that can change the pushdown
as described above. And then, it is possible for the final $\NCM$ machine built
to ``simulate'' the pushdowns (with the $\DFA$), at least in terms of the number of copies of each symbol. This allows the final $\NCM$ to count the number of times each transition symbol associated with increasing 
some counter $j$ is read, minus the number of transitions applied that decrease
counter $j$. Since the increasing and decreasing
transition symbols read can be intermixed (since the $\DFA$ is some permuted version
of the $\NPDA$), a separate counter is used for counting increasing transitions
applied and for counting decreasing transitions. An additional complication is that
the transitions applied by the $\NPCM$ $M$ depend on whether each counter $j$ is
empty or non-empty; therefore, this is built into the simulation as well using
additional features. Formally, the lemma is as follows:
\begin{lemma}
\label{accproof}
Let $M = (Q, \Sigma,\Gamma,\delta,q_0,F)\in \NPCM$ with $k$
counters 
which satisfies the conditions of Lemma \ref{NPCMnormalform}. 
For all
$q \in Q$, a machine $M' \in \NCM$ can be constructed 
such that $L(M') = \acc_M(q)$.
\end{lemma}
\begin{proof}
Let
$T$ be a set of labels in bijective correspondence with transitions
of $\delta$.

Let $X \subseteq \{1, \ldots , k\},  r,p \in Q, b,c \in \Gamma$,
either $b = c = Z_0$, or $b \neq Z_0, c \neq Z_0$. 
First, create an intermediate $\NPDA$ $M_{r,p,X,b,c}$ over 
input alphabet $\Delta = T \cup \{ 1, \ldots, k\}$. The
input to this machine is of the form
\begin{equation}
\label{NPDAinput} z = t_1 y_1 \cdots t_n y_n,
\end{equation}
$t_i \in T$, $y_i$ is a (possibly empty) subsequence\footnote{A subsequence of a string
is any string that can be obtained by deleting characters arbitrarily from any set of positions.} of
$1 \cdots k$, for $1 \leq i \leq n$, and each number of $\{1, \ldots, k\}$ appears at most
once in $z$. Intuitively, all numbers $j$ in the set
$X$ enforce that only transitions on counter $j$ being
non-empty are applied until the number $j$ appears in $z$ 
(if $z$ occurs), at which point, only transitions on counter
$j$ being zero are applied. Thus, the number $j$ being read as part of the input is
a guessed ``trigger'' that indicates counter $j$ is empty.

Then, $M_{r,p,X,b,c}$ accepts all words of the form
of $z$ in Equation (\ref{NPDAinput}), where 
$\alpha = t_1, \ldots, t_n$ and:
\begin{enumerate}

\item $r \xrightarrow{Q}_{\alpha} p$,
\item $b \xrightarrow{\Gamma}_{\alpha} c$,

\item for each $j$, $1 \leq j \leq k$, one of the following is true:
\begin{itemize}
\item $j \notin X$, $t_1, \ldots, t_n$ are all defined on counter $j$ being $0$,
 and $j$ does not appear in $z$,
\item $j \in X$,
 there exists $l$ such that $1 \leq l \leq n$, 
$t_1, \ldots, t_l$ are defined on counter $j$ being positive, $t_{l+1}, \ldots, t_n$ are all defined on counter $j$ being $0$, where $j$ can only (optionally) appear in $y_l$ and in no other position of $z$, with ($l<n$ implies $y_l$ contains $j$),
and
($j$ is in $y_l$ implies $t_l$ decreases counter $j$).
\end{itemize}

\end{enumerate}

Notice that in point 3, the only way for $j$ to be in $X$ but not in $z$ is for $l$ to be equal to $n$ (if $l=n$ then either $j$ is
in $y_l$ or not).
Thus, $M_{r,p,X,b,c}$ can read a word
of the form of $z$, verifying condition 2 using the pushdown by
simulating the pushdown of $M$ faithfully, and verifying conditions
1 and 3 using the finite control.

The $\NPDA$ accepts sequences of transition
labels of $M$ of the form of Equation (\ref{NPDAinput}) that start
with a letter $b$ on a pushdown, and without reducing the size
of the pushdown below the $b$,
eventually returning to a pushdown with $c$ in place of $b$ and nothing
above it. This machine also enforces using the finite
control that, for each counter $j$, if $j \in X$, then it
simulates transitions on counter $j$ being positive until it (optionally) reads
$j$ on the input, then only transitions on $j$
being zero are applied. Also, by condition \ref{disjointreversal} of the normal form
of Lemma \ref{NPCMnormalform}, if some transition label decreases counter $j$, no
further transition label can increase counter $j$. But, the machine never counts the number of increase or decrease transitions.
However, $L(M_{r,p,X,b,c})$ is a context-free language, and 
it is known that, for 
every context-free language, there is a
regular language with the same Parikh map \cite{harrison1978} since every context-free language is semilinear. Hence, it is possible
to construct a $\DFA$ $M'_{r,p,X,b,c}$ accepting a language with the
same Parikh map.

Next, the final $M' = (Q', \Sigma, \Gamma, \delta',q_0, F') \in \NCM$ to accept $\acc_M(q)$ with $2k$ counters will be constructed.
The first $k$ counters simulate the increases of the $k$
counters of $M$, and the next $k$ are increased for every
decrease of the corresponding counter of $M$. The states $Q'$
of $M'$ include $q_0$ (the initial state of $M$) and $Q \times \Gamma \times 2^{\{1, \ldots, k\}}$,
plus certain states of the simulated $\DFA$s of $M'_{r,p,X,b,c}$ to be
described next. The first component of states in $Q \times \Gamma \times 2^{\{1, \ldots, k\}}$ stores the simulated state
of $M$, the second contains the top symbol of the simulated pushdown, and the third contains a nondeterministic guess as to which of the counters are currently non-empty (this is verified to be correct at the end of the computation).
$M'$ starts in state $q_0$, and $F' = \{q_f\}$, where $q_f$ is a new state.

Inputs to $M'$ are of the form 
$qu c_1^{\alpha_1} \cdots  c_k^{\alpha_k}$,
$u\in \Gamma^*, \alpha_1, \ldots, \alpha_k \in \mathbb{N}_0$.
Next, the transition function $\delta'$ of $M'$ will be defined.
First, $q_0$ switches immediately to $(q_0', Z_0, \{1, \ldots, k\})$
($q_0'$ is defined in Lemma \ref{NPCMnormalform}) while increasing each of the first $k$
counters to $1$, and keeping the last $k$ counters as zero as $M$ immediately increases
all counters to $1$ from $q_0$ to $q_0'$.
Let $t$ be a transition of $M$ that pushes on the pushdown, which is of the form:
\begin{equation}
(p, b c, z_1, \ldots, z_k) \in \delta(r,a,b,x_1,\ldots, x_k),
\label{inDFA}
\end{equation}
where $r,p \in Q, a \in \Sigma \cup \{\lambda\}, 
b,c \in \Gamma, x_1, \ldots, x_k \in \{0,1\}$. Let $X = \{j \mid x_j = +1\}$ (those counters that are non-empty when this transition is applied).
Then create 
transitions
$$((p,c,X'), z_1', \ldots, z_k', z_1'', \ldots, z_k'') \in \delta'((r,b,X),b,1, \ldots, 1,y_1, \ldots, y_k),$$ for all $y_1, \ldots, y_k \in \{0,1\}$ ($b$ is the input letter here since the store language of $M$ is being accepted, and there is no pushdown of $M'$), and $X'$ whereby 
\begin{enumerate}
\item $z_j' = \begin{cases} +1 & \mbox{if~} z_j = +1,\\ 0 & \mbox{otherwise,}\end{cases}$
\item $z_j'' = \begin{cases} +1 & \mbox{if~} z_j = -1,\\ 0 & \mbox{otherwise,}\end{cases}$
\item \label{eqn3} $\{j \mid y_j=0\} \subseteq X$,
\item $(X'' := \{j \mid j \in X \mbox{~and~} z_j \in \{0,+1\}\}) \subseteq
X' \subseteq X.$
\end{enumerate}
Also, for all $r \in Q, b \in \Gamma,
y_j \in \{0,1\}, X \subseteq \{1,\ldots, k\}$ such that $\{j \mid y_j =0\} \subseteq X$, create transitions of $M'$ from
\begin{equation}
\delta'((r,b,X),\lambda,1, \ldots, 1,y_1, \ldots, y_k),
\end{equation}
that allow $M'$ to nondeterministically
guess a word $v \in \Delta^*$, and simulate the reading of such words but with $\lambda$ transitions by $M_{r,p,X,b,c}'$
for some $p \in Q, c \in \Gamma$,
representing transitions of $M$ starting with exactly the non-empty counters
in $X$. During this simulation, $M'$ uses states of the form 
\begin{equation}
(r,b,X,q',Y),
\label{newstates}
\end{equation} where $r,b,X$ do not change, $q'$ is the simulated state of 
$M_{r,p,X,b,c}'$, and $Y \subseteq \{1, \ldots, k\}$.
During the simulation of $v$, if this reads a transition label $t'$
representing a transition of $M$ where counter $j$ increases,
then $M'$ adds $1$ to counter $j$.
If $t'$ represents a transition where counter $j$ decreases, then $M'$ adds $1$ to counter $j+k$.
All numbers $j$ in $v$ read are added to
the set in the last component in Equation (\ref{newstates}).
If the simulated machine
$M_{r,p,X,b,c}'$ accepts $v$ by being in state $(r,b,X,q',Y)$, where $q'$ is a final state of
$M_{r,p,X,b,c}'$,
then $M'$ switches to state $(p,c,\overline{X})$, where 
\begin{equation}
\label{eq4}
\overline{X} =  X - Y.
\end{equation}

For $M'$ to accept, when $M'$ is in some state $(q,b,X)$ ($M'$ can stop when simulating
a computation of $M$ in state $q$) for some $b \in \Gamma$ and $X$, 
$M'$ reads $b$ from the input  and guesses that there are no more letters from $\Gamma$
on the input. At this point, it 
subtracts the value of counter $j+k$ from counter $j$, for all $j$, $1 \leq j \leq k$, and then verifies that the rest of the input is $c_1^{i_1} \cdots c_k^{i_k}$, 
where the counters
are $(i_1, \ldots, i_k, 0, \ldots, 0)$, and that $X = \{l \mid i_l >0\}$, before switching to the final state $q_f$.

Consider a computation of $M$ ending in state $q$ starting at the initial configuration,
\begin{equation}
(q_0, w_0, u_0, i_{0,1}, \ldots, i_{0,k}) \vdash_M \cdots 
\vdash_M (q_n, w_n, u_n, i_{n,1}, \ldots, i_{n,k}),
\label{derivation}
\end{equation}
where $q_1 = q_0', u_0 = u_1 = Z_0, i_{0,1} = \cdots = i_{0,k} = 0, i_{1,1} = \cdots = i_{1,k} = 1$ (by the normal form), and $q = q_n$.
Let $m = |u_n|$. It will be shown that $M'$ can accept
$q u_n c_1^{i_{n,1}} \cdots c_k^{i_{n,k}}$.
For each $l$ from $1$ to $m$, there must be some maximal
configuration $z_l$, $1 \leq z_l \leq n$, where $|u_{z_l}| = l$.
Then notice that, for each $l$, for all $x > z_l$, $|u_x| > l$; i.e.\
if the pushdown went down in size, it would need to return to size $l$ again, contradicting maximality. Consider the computation between
$(q_{z_l+1}, w_{z_l+1}, u_{z_l+1}, i_{z_l+1,1}, \ldots, i_{z_l+1,k})$
and $(q_{z_{l+1}}, w_{z_{l+1}}, u_{z_{l+1}}, i_{z_{l+1},1}, \ldots, i_{z_{l+1},k})$, for $1 \leq l < m$
(for the pushdown visualized in Figure \ref{fig1}, from the dot to the right of each blue dot until the next blue dot).
This computation does not reduce the size of the pushdown, by
the maximality, and therefore one only needs to consider replacing
the top of the pushdown, changing states, and counters appropriately.
This can be calculated by simulating
$M'_{q_{z_l+1},q_{z_{l+1}},X,b,c}$, where $b$ is the top symbol
of $u_{z_l+1}$, and $c$ is the top symbol of $u_{z_{l+1}}$,
and $X$ is the set of non-empty counters in the first configuration. This simulation
can nondeterministically guess 
a word $v\in \Delta^*$ letter-by-letter using $\lambda$ transitions,
giving the correct number
of counter increases and decreases for each counter $j$, which
are added to counters $j$ and $k+j$ of $M'$, respectively. For all counters $j$ emptied during this part of the computation of $M$, the simulated machine can read the number $j$ as well. Furthermore, the machine 
$M_{q_{z_l+1},q_{z_{l+1}},X,b,c}$ (before permuting from $\NPDA$ to $\DFA$) will 
enforce that once $j$ is read, only transitions on counter
$j$ being empty can occur. Thus, $M'$ can continue with the correct
counter values and in state $(q_{z_{l+1}},c,\overline{X})$, where
$\overline{X}$ is obtained from $X$ as per Equation (\ref{eq4}). If $l+1<m$, then $c$
can be read from the input since it must be the symbol at position
$l+1$ of the input $u_n$ using a transition
created in Equation (\ref{inDFA}) to $(q_{z_{l+1}+1},d,X')$, where $d$ is at position
$l+2$ of $u_{z_{l+1}+1}$, and $X'$ removes all those counters $j$ from
$X$ that subtracted to zero in Equation (\ref{derivation}). The rest follows inductively
until the final configuration, in some state $(q,c,X)$, where $c$ is the last symbol
of $u_n$. Then because each counter $j$ was removed from the set $X$ exactly
when and if counter $j$ emptied, subtracting counter $k+j$ from $j$, for all 
$1\leq j \leq k$ will give $i_{n,j}$.

For the converse, there is an accepting computation of
$ \gamma = q u c_1^{i_1} \cdots c_k^{i_k} \in L(M')$, that immediately switches
from $q_0$ with $0$'s on the counters to $(q_0', Z_0, \{1,\ldots,k\})$ with $1$ on the first $k$ counters
and $0$ on the rest; then the rest of the computation is
as follows, for all $j$, $0 \leq j <n$:
$$(p_j, \gamma_j, i_{j,1}, \ldots, i_{j,2k}) \vdash_{M'}^* (p_j', \gamma_j', i_{j,1}', \ldots, i_{j,2k}') \vdash_{M'} 
(p_{j+1}, \gamma_{j+1}, i_{j+1,1}, \ldots, i_{j+1,2k}),$$
and $(p_n, \gamma_n, i_{n,1}, \ldots, i_{n,2k}) \vdash_{M'}^* (q_f, \lambda, 0, \ldots, 0)$, where $p_0 = (q_0', Z_0, \{1, \ldots, k\}),
\gamma_0 = q u c_1^{i_1} \cdots c_k^{i_k}, i_{0,l} = 1, 1 \leq l \leq k, i_{0,l} = 0, k+1 \leq l \leq 2k$, between $p_j$ and $p_j'$ can either be empty, or 
 one of the $\DFA$s  constructed above is simulated, between $p_j'$ and $p_{j+1}$ reads an input letter of $u$, $\bar{p_n} = q$, and
$\gamma_n = c_1^{i_1} \cdots c_k^{i_k}$.
Let $p_j = (\bar{p_j}, b_j, X_j), p_j' = (\bar{p_j}', d_j, X_j'), \theta_{j,l} = i_{j,l}-i_{j,l+k}, 
\theta_{j,l}' = i_{j,l}'-i_{j,l+k}', 0 \leq  j \leq n, 1 \leq l \leq k$.

Then for each $j$, $0 \leq j < n$, the $\DFA$ $M_{\bar{p_j},\bar{p_j}', X_j, b_j, d_j}'$ is simulated guessing a word $v \in \Delta^*$
letter-by-letter on $\lambda$ transitions. This word is a permutation of a word accepted by the corresponding $\NPDA$ that accepts sequences of transitions $v'$ of $M$ that,
starting at state $\bar{p_j}$ with $b_j$  on the top of the stack, can eventually (without ever reducing the size of the
pushdown past this point and so without affecting what is below), end up with $d_j$ replacing $b_j$ and leaving the
rest of the pushdown unchanged, increasing counter $l$ for every
transition doing so read, and increasing counter $l+k$ for every decrease of counter $l$. Then, a transition that reads $d_j$ is 
applied, and if $j<n-1$, then a push transition of $M$ can be simulated.

Hence, $(q_0, y_0, Z_0, 0, \ldots, 0) \vdash_M^* (\bar{p_n}, y_n, d_0 \ldots d_n, \theta_{n,1}, \ldots, \theta_{n,k}), \bar{p_n} =q, y_n = \lambda$, for some 
$y_0, \ldots, y_n \in \Sigma^*$. Thus, $\gamma \in \acc_q(M)$.
\qed
\end{proof}

\begin{lemma} \label{coacclemma}
Let $M = (Q, \Sigma,\Gamma,\delta,q_0,F)\in \NPCM$ with 
$k$ counters which satisfies the conditions of Lemma \ref{NPCMnormalform}. For all
$q \in Q$, $M' \in \NCM$ can be constructed such that $L(M') = \coacc_M(q)$.
\end{lemma}
\begin{proof}
For each $p \in F$, take $M$ and construct $\NPCM$ $M_p^R$, constructed similar
to the
standard reversal construction. First, $M_p^R$ guesses an arbitrary pushdown and counter contents and pushes them. Then, it simulates $M$ ``in reverse'' starting at initial state $p$ and ending in final
state $q$. That is, counter decreases instead increase, and vice versa. 
And, if $M$ has $x$ at the top
of the pushdown and replaces it with $y \in \Gamma$, $M_p^R$
replaces $y$ with $x$, if $M$ replaces $x$ with
$yz, y,z \in \Gamma$, then $M_p^R$ pops $z$ (with a new intermediate state), and then
replaces $y$ with $x$, and if $M$ replaces $x$ with $\lambda$ then $M^R$ 
pushes $x$ on every pushdown letter.
Then, the union of $g(\acc_{M_p^R}(q))$ over all $p \in F$ is equal
to $\coacc_M(q)$, where $g$ is a gsm that does not output on the new intermediate
states. Furthermore, $\LFam(\NCM)$ is closed under gsm mappings and union. The lemma follows.
\qed
\end{proof}

\begin{proposition}
\label{StoresOfNPCM}
If $M$ is an $\NPCM$, then $S(M)  \in \LFam(\NCM)$. Thus,
$\SFam(\NPCM) \subseteq \LFam(\NCM)$.
\end{proposition}
\begin{proof}
First, let $M$ be an $\NPCM(k)$ with state set $Q$
satisfying the conditions of Lemma \ref{NPCMnormalform}.
Then $\acc_M(q),\coacc_M(q) \in \LFam(\NCM)$ for each $q \in Q$.
Since 
$S(M) = \bigcup_{q\in Q} \acc_M(q) \cap \coacc_M(q)$ and 
$\LFam(\NCM)$ is closed under intersection and union
\cite{Ibarra1978}, it is immediate that $S(M) \in \LFam(\NCM)$.

By Lemma \ref{NPCMnormalform}, this must be true for all $\NPCM$ machines,
since $\LFam(\NCM)$ is closed
under gsm mappings.
\qed
\end{proof}


It is worth noting that even though $\NPDA$s alone only produce regular store languages, and $\NCM$s produce $\NCM$ store languages, this is not enough to immediately conclude 
that machines combining a store that only produces regular store languages with counters produce store languages in
$\LFam(\NCM)$. For example, stack automata produce regular store languages,
but machines combining a stack plus reversal-bounded counters produce store
languages not in $\LFam(\NCM)$ \cite{StoreLanguages}. In this case though, restricting
the stack to also be reversal-bounded produces store languages in $\LFam(\NCM)$.

The following proposition shows, in some sense, a converse to Proposition \ref{StoresOfNPCM}.

\begin{proposition}
\label{converseNPCM}
Let $M \in \NCM$. Then there exists $M' \in \NPCM$ where the pushdown is zero-reversal-bounded, and a fixed word $fZ_0$ such that
$(fZ_0)^{-1} S(M') = L(M)$. Similarly with $M' \in \NQCM$,
with the queue zero-reversal-bounded.
\end{proposition}
\begin{proof}
Let $M \in \NCM$. Assume without loss of generality that
$M$ immediately leaves the initial state $q_0$ without re-entering it,
and switches to a unique final state $f$ only with all counters empty.
Create an $\NPCM$ machine $M'$ that  
on input $w \in \Sigma^*$, copies $w$ to the
pushdown (which starts with bottom of stack marker $Z_0$),
while in parallel simulating the computation of $M$ on $w$ with the counters, accepting in state $f$
with the counters empty.
Then $M'$ is in state $f$ with $Z_0w$ in the pushdown and all counters are empty if and only if $w \in L(M)$. Hence, 
$(fZ_0)^{-1} (S(M')) = L(M) $.

Since a zero-reversal-bounded queue operates identically to a pushdown, the result for $\NQCM$ follows.
\qed \end{proof}


Thus, even though the store languages of nondeterministic pushdown automata are all regular, and it is known that the store languages of
$\NCM$ are all $\LFam(\DCM)$ \cite{StoreLanguages}, the store languages of $\NPCM$ combining a pushdown storage with multicounter stores is more general
than $\DCM$.

\begin{corollary}
\label{NPCMCor}
$\LFam(\NCM)$ is the smallest family of languages containing 
$\SFam(\NPCM)$ that is closed under left quotient with words.
\end{corollary}

Next, the extension from $\NPCM$ to a new model, $r$-flip $\NPCM$ is studied.
As mentioned in Section \ref{sec:prelims}, an $r$-flip $\NPCM$ is an $\NPCM$ with an additional instruction to flip the pushdown that can applied up to $r$ times (in addition to having reversal-bounded counters). 
A flip transforms the contents of the pushdown from $Z_0 \gamma$ to $Z_0 \gamma^R$ (thus, the bottom-of-stack marker stays in place).
The proof uses the newly proven Proposition \ref{StoresOfNPCM}
that showed that store languages of $\NPCM$ are in $\LFam(\NCM)$.
Notice that such a model is more powerful than $r$-flip $\NPDA$s, as it is known that $\{a^n b^n c^n \mid n \geq 1\}$ cannot
be accepted \cite{FlipPushdownReversals}, whereas it can with a $k$-flip $\NPCM$ (or even a two counter $\DCM$).  

Let $l \geq 0$. For such a machine $M$ with state set $Q$ and $ q\in Q$, let $\acc_{q,l}(M)$ 
be the set of all configurations $q \gamma c_1^{i_1} \cdots c_k^{i_k}$ that are
 reachable from the initial configuration by using at most $l$ pushdown flips, and let $\coacc_{q,l}(M)$ be the set of all configurations
 $q \gamma c_1^{i_1} \cdots c_k^{i_k}$ that can reach a final configuration with at most $l$ pushdown flips.

\begin{lemma}
\label{firstdirectionflip}
If $M = (Q,\Sigma,\Gamma,\delta,q_0,F)$ is an $r$-flip $\NPCM$ with $k$ counters, then for all 
$l \leq r$,  $M' \in \NCM$ can be constructed such that $L(M') = \acc_{q,l}(M)$.
\end{lemma}
\begin{proof}
Let $q \in Q$. First, it is clear that without any flips, then the machine operates like a normal $\NPCM$ as no flips are applied (and can therefore be omitted). Therefore, $\acc_{q,0}(M) \in \NCM$ by Lemma \ref{accproof}. 

Briefly, a finite-crossing $\NCM$ is a two-way $\NFA$ augmented by reversal-bounded counters, such that there is a bound on the number of times the boundary between two adjacent input cells is crossed
\cite{Gurari1981220}. It is known that finite-crossing $\NCM$s are equivalent to (one-way) $\NCM$s
\cite{Gurari1981220}.

Next, a procedure to construct $\acc_{q,1}(M)$ will be described. A similar procedure can iterate up to any $l$. Consider any transition $t$ of $M$ where $t$ flips the pushdown,
$(s,{\rm flip},z_1, \ldots, z_k) \in \delta(p,a,b,x_1, \ldots, x_k), p,s \in Q, a \in \Sigma \cup \{\lambda\}, b \in \Gamma, x_j \in \{0,1\}, z_j \in \{-1,0,+1\}, 1 \leq j \leq k$.  An intermediate finite-crossing $\NCM$ $M_t$ will be built that accepts all reachable configurations involving one flip, $t$, which is the final transition applied in $M_t$ as follows: Consider the $\NCM$ machine $M'$ accepting $\acc_{p,0}(M)$. Then, $M_t$, on input of the form 
\begin{equation}
s Z_0 \gamma c_1^{i_1} \cdots c_k^{i_k},
\label{Mtequation}
\end{equation}
verifies that $p Z_0 \gamma^R c_1^{i_1'} \cdots c_k^{i_k'}$ is in $L(M')$, where the last letter of $\gamma^R$ is $b$, $x_j =1$ if and only if $i_j'>0$, and $i_j = i_j' + z_j$, for all $j$, $ 1\leq j \leq k$. This can be done, as $M_t$ is finite crossing, and can therefore read $\gamma$ in reverse. Indeed, $M_t$
accepts all strings of the form Equation (\ref{Mtequation}) that only flip once, via transition $t$, on the last transition of the computation. Since finite-crossing $\NCM$s can be converted to an $\NCM$ accepting the same language \cite{Gurari1981220},
it is possible to build a (one-way) $\NCM$ accepting $L(M_t)$; call this $M_t'$.

Next, build an intermediate $\NPCM$ $M_t''$ that nondeterministically guesses some word (\ref{Mtequation}) accepted by $M_t'$ and puts it in its stores (it does this using a certain set of states $Q'$ disjoint from $Q$), then it simulates $M$ with no flips. That is, it pushes $\gamma$ on the pushdown and puts
$i_j$ in counter $j$ for $ 1 \leq j \leq k$, and switches to state $s$ if (\ref{Mtequation}) is accepted by
$M_t'$ which can be verified only with counters. From there, $M_t''$
continues the simulation of $M$ starting in state $s$ without any flips. It is clear that the reachable configuration to a state in $Q$ are exactly the reachable configurations of $M$ with one flip via transition $t$ after $t$ has been applied. Therefore, by Proposition \ref{StoresOfNPCM}, the reachable configurations of $M_t''$ (using any state in $Q$, ignoring $Q'$) can be accepted by an $\NCM$. Since $\NCM$ is closed under union (over all transitions), it is straightforward to show that $\acc_{q,1}(M)$ can be accepted by an $\NCM$, for all $q$. Similarly for $\acc_{q,l}(M)$,
for any $l \geq 0$.
\qed \end{proof}

\begin{lemma} 
\label{seconddirectionflip}
If $M$ is an $r$-flip $\NPCM$ with $k$ counters and state set $Q$, then for all $l \leq r, q \in Q$, $M' \in \NCM$ can be constructed such that 
$L(M') = \coacc_{q,l}(M)$.
\end{lemma}
This is similar to the proof of Lemma \ref{coacclemma}.

Then we can conclude the following:
\begin{proposition} If $M$ is an $r$-flip $\NPCM$, then $S(M) \in \LFam(\NCM)$. Thus,
$\SFam(r\mbox{-flip }\NPCM) \subseteq \LFam(\NCM)$.
\label{flipstore}
\end{proposition}
\begin{proof}
An $\NCM$ can be built accepting $\acc_{q,l}$ and $\coacc_{q,l}$, for all $q$ and $l$, from $0 \leq l \leq r$ by Lemmas \ref{firstdirectionflip} and \ref{seconddirectionflip}. Then $S(M) = \{x \mid x \in \acc_{q,l} \cap \coacc_{q,l'},\mbox{~where~} l + l' \leq r\}$. Since $\NCM$ is closed under intersection and union, the proof follows.
\qed \end{proof}

One interesting subfamily of $\NPCM$ is machines
where the pushdown is a counter. That is, the machines have one
unrestricted counter without a reversal-bound, plus some number of reversal-bounded counters.
Call this type of machine $\NCACM$. As mentioned in Section \ref{sec:prelims}, a machine with
two unrestricted counters has the same power as a Turing machine
\cite{HU}.
In \cite{StoreLanguages}, it was shown that the store language of every
$\NCM$ is actually a deterministic $\NCM$ machine ($\DCM$). Next
it will be shown that this is also true of $\NCACM$, which will follow quite
easily from the proof that the store languages of all $\NPCM$ machines
are in $\LFam(\NCM)$.
\begin{proposition}
\label{allInDCM}
$\SFam(\NCACM) \subseteq \LFam(\DCM)$.
\end{proposition}
\begin{proof}
By Proposition \ref{StoresOfNPCM}, all store languages of 
$\NCACM$ machines are in $\NCM$, as $\NCACM$ is a special type
of $\NPCM$. But, notice that the store language of a
$\NCACM$ is a bounded language (if there are $k$ counters, then the store language is
a subset of $c_1^* \cdots c_k^*$). Furthermore, it is known
that every bounded $\NCM$ language is a $\DCM$ language
\cite{IbarraSeki}. 
\qed \end{proof}

From this result, the strong result is obtained, that is is possible
to test equality or even containment between the store languages
of two $\NCACM$ machines.
\begin{proposition}
It is decidable, given $M_1, M_2 \in \NCACM$, whether
$S(M_1) = S(M_2)$, and
$S(M_1) \subseteq S(M_2)$.
\end{proposition}
\begin{proof}
It follows from Proposition \ref{allInDCM} that
the store language of both $M_1$ and
$M_2$ are in $\DCM$. It is also known that equality and containment
are decidable for $\DCM$ \cite{Ibarra1978}.
\qed \end{proof}

Results on store languages are summarized in Table \ref{storeresults}, together with where the result was shown. The models below the line have reversal-bounded counters attached.
\begin{table}
\caption{Each machine model in column 1 (above the line are models without counters) has store languages contained in the family in column 2, with the result shown in column 3.}
\begin{center}
\begin{tabular}{l | l | l }
machine model & store languages in & proven in \\ \hline
$\NPDA$ & $\LFam(\REG)$ & \cite{GreibachCFStore} \\
$r$-flip $\NPDA$ & $\LFam(\REG)$ & \cite{StoreLanguages} \\
$\NSA$ & $\LFam(\REG)$ & \cite{KutribCIAA2016}\\
$\NRBQA$ & $\LFam(\REG)$ & \cite{StoreLanguages} \\
$\NRBTA$ & $\LFam(\REG)$ & \cite{StoreLanguages} \\ \hline
$\NCM$ & $\LFam(\DCM)$ & \cite{StoreLanguages} \\
$\NCACM$ & $\LFam(\DCM)$ &  Proposition \ref{allInDCM}\\
$\NPCM$ & $\LFam(\NCM)$ & Proposition \ref{StoresOfNPCM} \\
$r$-flip $\NPCM$ & $\LFam(\NCM)$ &  Proposition \ref{flipstore}\\
$\NRBSCM$ & $\LFam(\NCM)$ & \cite{StoreLanguages}\\
$\NRBQCM$ & $\LFam(\NCM)$ & \cite{StoreLanguages}\\
$\NRBTCM$ & $\LFam(\NCM)$ & \cite{StoreLanguages}
\end{tabular}
\end{center}
\label{storeresults}
\end{table}%

\section{Connections Between Store Languages and Reachability Problems}
\label{sec:reachable}

The $\pre_M^*$ and $\post_M^*$ operators are commonly studied in the area of model checking and reachability. In particular, it is known that for a $\NPDA$ $M$ and a regular language $C$,
$\pre^*_M(C)$ and $\post^*_M(C)$ are in $\LFam(\REG)$ \cite{PushdownVerification}. Also,
for $M \in \NCM$ and $C \in \LFam(\NCM)$, it is known that $\pre^*_M(C)$ and $\post^*_M(C)$ are in $\LFam(\DCM)$ \cite{IbarraSu}.

In this section, a connection is made between store languages and the $\pre_M^*$ and $\post_M^*$ operators. The first
direction is essentially immediate.
\begin{proposition}
Let $\MFam$ be a machine model, and let $M \in \MFam$. Then
the store language of $M$, $S(M) = \post_M^*(c_0) \cap \pre_M^*(C_f)$, where $c_0$ is the initial configuration of $M$, 
and $C_f$ is the regular set of final configurations of $M$. 	
\label{sameregular}
\end{proposition}
Here, $C_f$ is considered regular since $\conf(M)$ is defined to be regular and acceptance is always by final state \cite{StoreLanguages}, and therefore
$\conf(M)$ is simply restricted to start with the final state set.

\begin{corollary}
Let $\MFam$ be a machine model. If a regular set of configurations $C$, $C \subseteq \conf(M), M \in \MFam$ implies both
$\pre_M^*(C)$ and $\post_M^*(C)$ are regular, then $\SFam(\MFam) \subseteq \LFam(\REG)$.
Also, if a regular set $C$ implies $\pre_M^*(C)$ and $\post_M^*(C)$ are in $\LFam(\NCM)$ ($\LFam(\DCM)$ respectively), then 
$\SFam(\MFam) \subseteq \LFam(\NCM)$ ($\SFam(\MFam) \subseteq \LFam(\DCM)$ respectively).
\label{connectiontoprepost}	
\end{corollary}
This is true immediately by Proposition \ref{sameregular}, since $c_0$ and $C_f$ are regular, and $\LFam(\REG),\LFam(\NCM),\LFam(\DCM)$ are closed
under intersection \cite{HU,Ibarra1978}, and they all contain $\LFam(\REG)$.

There is also a converse of sorts to Corollary \ref{connectiontoprepost} but it is slightly more complicated. First, definitions are required. 
Consider a machine model $\MFam$. A set of configurations $C$
can be {\em loaded} by $\MFam$ if, for all $M = (Q,\Sigma,\Gamma,\delta,q_0,F) \in \MFam$ with $C \subseteq \conf(M)$, there is a machine $M' \in \MFam$ with state set $Q' \supseteq Q$ that, on input 
$q \gamma \$ x\$$ where $c = q\gamma \in \conf(M), q \in Q, \gamma \in \Gamma^*, x \in \Sigma^*$, and $\$$ is a new symbol, operates as follows: $M'$ reads $c$ while using states in $Q' - Q$, and changes its store configuration to $c$ 
only switching to a state of $Q$ ($q$ specifically) after reading $\$$ if $c \in C$, and switches to a unique state $q_N \in Q' - Q$ 
if $c \notin C$. Then, from $c$, $M'$ simulates $M$ on $x$, accepting if it reads the whole input. Here $M'$ is called a $C$-loaded version of $M$. Notice that $M'$ simulates $M$ on $x$ if and only if $c \in C$, as $M$ is only defined on states of $Q$.
It is said that $\MFam$ can be loaded by sets from some family $\LFam$ if,
for all $M\in \MFam$ and $C \subseteq \conf(M)$ with $C \in \LFam$, then $C$ can
be loaded by $\MFam$.

\begin{proposition} Let $\MFam$ be any machine model that can be loaded by sets from some family $\LFam$.
For all $M \in \MFam$, $C \subseteq \conf(M)$, then $\post_M^*(C) = S(M') \cap \conf(M)$, where $M'$ is the $C$-loaded version of $M$.
\end{proposition}
\begin{proof}
Let $M' \in \MFam$ be a $C$-loaded version of $M \in \MFam$.

Let $c \in \post_M^*(C)$. Then $c \in \conf(M)$. Then there exists $c_0, \ldots,c_i, i \geq 0$ such
that $c_0 \Rightarrow_M c_1\Rightarrow_M \cdots \Rightarrow_M c_i$, where $c_0 \in C, c = c_i$, say on input word $x \in \Sigma^*$. Then, on input $c_0 \$ x \$$, after reading $c_0$,
$M '$ is in configuration $c_0$, and after reading $x\$$ can be in configuration
$c$, which is accepting in $M'$. Thus, $c \in S(M')$.

Let $c \in S(M') \cap \conf(M)$. Since $c \in \conf(M)$ (i.e.\ it must use
some state of $Q$ and not $Q' - Q$), there must be some computation
of $M'$ whereby $M'$ reads some $c_0 \in C$, and is in configuration $c_0$,
and then eventually switches into configuration $c$ (since $c \in S(M')$). Thus,
$c \in \post_M^*(c_0)$.
\qed
\end{proof}

The following is immediate since $\conf(M)$ is always a regular language.
\begin{corollary}
\label{postcor}
Let $\MFam$ be a machine model that can be loaded by sets of configurations from $\LFam_1$,
and	let $\LFam_2$ be a family closed under intersection with regular languages. If 
$\SFam(\MFam) \subseteq \LFam_2$, then $\post_M^*(C) \in \LFam_2$, for all $C \in \LFam_1, M \in \MFam$.
\end{corollary}
In this paper, most models studied have store languages either in $\LFam(\REG),\LFam(\NCM)$, or $\LFam(\DCM)$, all of which are closed under intersection with regular languages \cite{Ibarra1978}, so the following is pointed out specifically:
\begin{corollary}
Let $\MFam$ be any machine model that can be loaded by sets of configurations
from $\LFam$:
\begin{itemize}
\item if $\SFam(\MFam) \subseteq \LFam(\REG)$, then $\post_M^*(C) \in \LFam(\REG)$ for all $C \in \LFam$ and $M \in \MFam$,
\item if $\SFam(\MFam) \subseteq \LFam(\NCM)$, then $\post_M^*(C) \in \LFam(\NCM)$
for all $C \in \LFam$ and $M \in \MFam$,
\item if $\SFam(\MFam) \subseteq \LFam(\DCM)$, then $\post_M^*(C) \in \LFam(\DCM)$ 
for all $C \in \LFam$ and $M \in \MFam$.
\end{itemize}
\end{corollary}

Analogously,
consider a machine model $\MFam$.
Then a set of configurations $C$
can be {\em unloaded} by $\MFam$ if, for all $M = (Q,\Sigma,\Gamma,\delta,q_0,F) \in \MFam$ with $C \subseteq \conf(M)$, there is a machine $M' \in \MFam$ that, on 
input $c \$  x \$$, $c \in \conf(M), x \in \Sigma^*$, operates as follows: $M'$ reads $c$ using states not in $Q$, and upon reading $\$$, switches to configuration $c$, then it simulates $M$ on $x$, and upon reading $\$$, verifies that the current 
configuration of $M'$ is in $C$ and accepts only in this case. Here, $M'$ is called a $C$-unloaded version of $M$. It is said that $\MFam$ can be unloaded by sets
from some family $\LFam$ if, for all $M \in \MFam$ and $C \subseteq \conf(M)$
with $C \in \LFam$, then $C$ can be unloaded by $\MFam$.

\begin{proposition} Let $\MFam$ be any machine model that can be unloaded by sets from some family 
$\LFam$. For all $M \in \MFam$, $C \subseteq \conf(M)$, then $\pre_M^*(C) = S(M') \cap \conf(M)$, where $M'$ is a $C$-unloaded version of $M$.
\end{proposition}
\begin{proof}
Let $M' \in \MFam$ be a $C$-unloaded version of $M$.	

Let $c \in \pre_M^*(C)$. Then $c \in \conf(M)$. Also, there exists $c_0, \ldots, c_i, i \geq 0$ such that 
$c = c_0 \Rightarrow_M \cdots \Rightarrow_M c_i \in C$, say on input word $x$. Then on input
$c \$ x\$$, after reading $c$, $M'$ switches to configuration $c$, then after reading $x$, $M'$ can be in
configuration $c_i \in C$. Then $M'$ verifies $c_i \in C$ and accepts. Therefore, $c$ was
an intermediate configuration in an accepting computation and $c \in S(M')$.

Let $c \in S(M') \cap \conf(M)$. Then there must be some computation of $M'$ whereby $M'$ reads $c$
and immediately switches to it, and eventually switches
to $c'$, which is verified to be in $C$. Hence, $c \in \pre_M^*(c'), c' \in C$.
\qed
\end{proof}

The following is therefore immediate:
\begin{corollary}
\label{precor}
Let $\MFam$ be a machine model that can be unloaded by sets of configurations from $\LFam_1$,
and	let $\LFam_2$ be a family closed under intersection with regular languages. If 
$\SFam(\MFam) \subseteq \LFam_2$, then $\pre_M^*(C) \in \LFam_2$, for all $C \in \LFam_1, M \in \MFam$.
\end{corollary}

Again, more specifically:
\begin{corollary}
Let $\MFam$ be machine model that can be unloaded by sets of configurations
from $\LFam$:
\begin{itemize}
\item if $\SFam(\MFam) \subseteq \LFam(\REG)$, then $\pre_M^*(C) \in \LFam(\REG)$ for all $C \in \LFam$ and $M \in \MFam$,
\item if $\SFam(\MFam) \subseteq \LFam(\NCM)$, then $\pre_M^*(C) \in \LFam(\NCM)$
for all $C \in \LFam$ and $M \in \MFam$,
\item if $\SFam(\MFam) \subseteq \LFam(\DCM)$, then $\pre_M^*(C) \in \LFam(\DCM)$ 
for all $C \in \LFam$ and $M \in \MFam$.
\end{itemize}
\end{corollary}

Combining together Proposition \ref{sameregular} and Corollaries \ref{postcor} and \ref{precor} gives the following:
\begin{theorem}
\label{main}
Let $\MFam$ be any machine model that can be loaded and unloaded by sets $C$ from $\LFam_1$, with $\LFam(\REG) \subseteq \LFam_1$, and let $\LFam_2$ be a family closed under intersection with regular languages and intersection. Then, $\SFam(\MFam) \subseteq \LFam_2$ if and only if $\post_M^*(C) \in \LFam_2$ and $\pre_M^*(C) \in \LFam_2$, for all $C \in \LFam_1, M \in \MFam$.	
\end{theorem}

All of the machine models listed in
Proposition \ref{loadedunloaded} are known to have regular store languages \cite{StoreLanguages} and it will be shown that they can be loaded and unloaded by regular sets of configurations. Therefore, the following is obtained:
\begin{proposition}
\label{loadedunloaded}
Given $\MFam$, of any of the following types:
$$\NPDA, \NSA, r\mbox{-flip }\NPDA, \NRBQA, \NRBTA.$$
Then, $\MFam$ can be loaded and unloaded by regular configurations. Hence, for all $M \in \MFam$ and regular configuration sets $C$, then $\pre_M^*(C)$ and $\post_M^*(C)$ are regular.
\end{proposition}
\begin{proof}
First, consider $\NPDA$s. Given regular $C$ accepted by a $\DFA$ $M_C$, one can build a
$C$-loaded version $M'$ of an $\NPDA$ $M$ as follows: read $q \gamma$ while in parallel
verifying that it is in $L(M_C)$, and placing $\gamma$ on the pushdown, then switching to state $q$, then simulating $M$. Moreover, one can build a $C$-unloaded version $M'$ of
an $\NPDA$ $M$ , by reading $q \gamma$, placing $\gamma$ on the stack
before switching to $q$, simulating $M$ on
some input which can eventually take $M'$ to some configuration $c' = p \alpha$.
Then, to verify $c' \in C$, $M'$ simulates a $\DFA$ accepting $C^R$, which
must be regular since regular languages are closed under reversal \cite{HU}. Indeed,
the pushdown is popped one symbol at a time in reverse. 

With stack automata, configurations encode the position of the read/write head \cite{StoreLanguages}.
To unload configurations, this requires nondeterministically guessing the final position
of the read head and marking it when this symbol
is getting pushed to the stack, otherwise the proof is the same
as with pushdown automata.

The proofs are similar with all other machine models listed. The second statement follows by the first result, by Theorem \ref{main}, and by closure of the regular languages under
intersection.
\qed
\end{proof}

This is already known for $\NPDA$s \cite{PushdownVerification}, however, this provides an alternate immediate proof based on the store language result. But for all the other models, we believe that these are new results of interest to the area of verification. Some of these models are indeed quite
powerful. For example, stack automata can accept non-semilinear languages in contrast to $\NPDA$s.

To complete this section, the $\pre^*$ and $\post^*$ operators will be examined on the models augmented by counters in Section \ref{sec:store} via an application of Theorem \ref{main}. To start, we see that not only can the models be loaded and unloaded by regular languages,
but also by languages in $\LFam(\NCM)$.

\begin{proposition} \label{augbycounters} Let $k \geq 0$.
Let $\MFam$ be any of the following machine models:
$$\NRBSCM(k), \NRBQCM(k), \NRBTCM(k), \NPCM(k), r\mbox{-flip }\NPCM(k), r \geq 0.$$
Then $\MFam$ can be loaded and unloaded by any set of configurations $C \in \LFam(\NCM(l))$ with $k + l$ counters.
In addition, for all $M \in \MFam$ and configuration sets $C \in \LFam(\NCM)$, both $\pre_M^*(C)$ and $\post_M^*(C)$ are in $\LFam(\NCM)$.
\end{proposition}
\begin{proof}
	Consider a $\NPCM(k)$ machine $M = (Q,\Sigma,\Gamma,\delta,q_0,F)$, and consider a configuration
set $C \subseteq \conf(M)$ such that $C \in \LFam(\NCM(l))$. Then every word $c \in C$ is of the form
$c = q\gamma c_1^{i_1} \cdots c_k^{i_k}$, where $ q\in Q, \gamma \in \Gamma^*, i_1,\ldots, i_k \geq 0$. Then when reading a word
of this form, $M'$ places $\gamma$ in the pushdown, and $i_j$ in counter $j$, $1 \leq j \leq k$, while in parallel, verifying $c \in C$ by simulating an
$\NCM(l)$ accepting $C$ using counters $k+1, \ldots, k+l$. $M'$ can switch states while verifying $c \in C$. Then,
$M'$ continues the simulation using the pushdown and the first $k$ counters. To unload, construct an $\NCM(l)$ accepting
$C^R$, which is possible since they are closed under reversal \cite{Ibarra1978}, and then decrease the counters from the $k$th to
the first, then pop from the pushdown to verify that the current configuration is in $C$. 

Similarly with the other models.

The second statement follows from Theorem \ref{main}, and because it has been seen that the store languages of all these models are in $\LFam(\NCM)$.
\qed
\end{proof}
These results are also new and quite general, and follow from results on store languages.
Also, the configuration sets $C$ can be more general than regular configurations, possibly
describing some numerical conditions that can be expressed with $\NCM$s.

Lastly, accepting predecessor and successor configurations of $\NCACM$ machines is addressed.
Recall that a trio is any family of languages closed under $\lambda$-free homomorphisms, inverse homomorphism, and intersection with regular languages. Also, a family is semilinear if all the languages in it are semilinear.
\begin{proposition}
\label{NCACMverification}
Let $\LFam$ be any semilinear trio.
	Machines in $\NCACM$ can be loaded and unloaded by configuration sets in $\LFam$.
In addition, for all $M \in \NCACM$ and configuration sets $C \in \LFam$, then $\pre_M^*(C)$ and $\post_M^*(C)$ are in
$\LFam(\DCM)$, and can therefore also be accepted by a deterministic logspace bounded or polynomial time Turing machine.
\end{proposition}
\begin{proof}
Let $M \in \NCACM$ with $k$ counters. Then any $C \subseteq \conf(M)$ is a subset of $c_1^* \cdots c_k^*$, which is a bounded language.
In \cite{CIAA2016}, it was shown that every bounded language in any semilinear trio is in fact in $\LFam(\DCM)$. Clearly then,
$\NCACM$ machines can be loaded by $\DCM$ machines just as in the proof of the previous proposition.

The second statement follows from Theorem \ref{main} and closure of $\LFam(\DCM)$ under intersection \cite{Ibarra1978}.
\qed
\end{proof}
In the proposition above, it is quite surprising that $\LFam$ can be any semilinear trio. Many families form semilinear trios, such as the regular languages, context-free languages, $\NPCM$, $r$-flip $\NPCM$s, all models in Proposition \ref{augbycounters}, or others \cite{CIAA2016}. No matter which of these families is used to describe a set of configurations $C$, the result of $\pre_M^*(C)$ and $\post_M^*(C)$ must always be in $\LFam(\DCM)$, for $M \in \NCACM$.

Results for specific machine models are summarized in Table \ref{verificationresults}.
\begin{table}
\caption{For each machine model in column 1 (below the double line with counters) has, for regular languages $C$ (or $\LFam(\NCM)$ $C$ below the line), $\pre_M^*(C)$ and $\post_M^*(C)$ contained in the family in column 2, with the result shown in column 3.}
\begin{center}
\begin{tabular}{l | l | l } \hline\hline
machine model & $C \in \LFam(\REG), \pre_M^*(C)$/$\post_M^*(C)$ in & proven in \\ \hline
$\NPDA$ & $\LFam(\REG)$ & \cite{PushdownVerification} \\
$r$-flip $\NPDA$ & $\LFam(\REG)$ & Proposition \ref{loadedunloaded} \\
$\NSA$ & $\LFam(\REG)$ & Proposition \ref{loadedunloaded}\\
$\NRBQA$ & $\LFam(\REG)$ & Proposition \ref{loadedunloaded} \\
$\NRBTA$ & $\LFam(\REG)$ & Proposition \ref{loadedunloaded} \\ \hline\hline
machine model & $C \in \LFam(\NCM), \pre_M^*(C)$/$\post_M^*(C)$ in & proven in \\ \hline
$\NCM$ & $\LFam(\DCM)$ & \cite{IbarraSu} \\
$\NCACM$ & $\LFam(\DCM)$ &  Proposition \ref{NCACMverification}\\
$\NPCM$ & $\LFam(\NCM)$ & Proposition \ref{augbycounters} \\
$r$-flip $\NPCM$ & $\LFam(\NCM)$ &  Proposition \ref{augbycounters}\\
$\NRBSCM$ & $\LFam(\NCM)$ & Proposition \ref{augbycounters}\\
$\NRBQCM$ & $\LFam(\NCM)$ & Proposition \ref{augbycounters}\\
$\NRBTCM$ & $\LFam(\NCM)$ & Proposition \ref{augbycounters}
\end{tabular}
\end{center}
\label{verificationresults}
\end{table}%

\section{Common Configurations}
\label{sec:common}

In this section, determining the common configurations between two machines will be briefly addressed. As discussed in Section \ref{sec:intro}, this has applications to problems of safety.

Given two machines $M_1, M_2$ from the same machine model $\MFam$ (with the same states, 
pushdown alphabet, and counter names in the case of $\NPCM$, or 
$r$-flip $\NPCM$, or $\NCM$), then the {\em common store configuration problem} is the
problem of determining whether there is a non-initial configuration between
$M_1$ and $M_2$ that can appear in an accepting computation of both.
Let
$$S(M_1,M_2) = \{ x \mid x \in S(M_1) \cap S(M_2), x \mbox{~non-initial in~} M_1 \mbox{~and~} M_2\}.$$
Thus, the common store configuration
problem is to determine whether $S(M_1,M_2) \neq \emptyset$.
Further, the {\em common store configuration infiniteness problem} is to determine whether  $S(M_1,M_2)$ is infinite.

Note that we assume that the states and the pushdown symbols are the same, and the counters match. However, one could also define the problem more generally, so that the two machines have a common configuration if there is some relabelling of the states, counters, and pushdown symbols of one machine that applied to a configuration of that machine gives a configuration of the second. However, it is possible to try every relabelling. Hence, if the common configuration problem is decidable, then it is
decidable for this more general problem as well.

First, a general decidability property is presented.
\begin{proposition}
\label{generalresult}
Let $\MFam$ be a machine model, and let $\LFam$ be a language family such that
\begin{itemize}
\item $\SFam(\MFam) \subseteq \LFam$ with an effective construction, 
\item $\LFam$ has a decidable emptiness problem, and
\item $\LFam$ is effectively closed under intersection. 
\end{itemize}
Then $S(M_1,M_2) \in \LFam$, and $\MFam$ has 
a decidable common store configuration problem.
Furthermore, if $\LFam$ additionally has a decidable infiniteness problem, then $\MFam$ has 
a decidable common store configuration infiniteness problem.
\end{proposition}
\begin{proof}
By the assumption, given machines $M_1, M_2 \in \MFam$, then
it is possible to build $S(M_1), S(M_2) \in \LFam$. Hence,
$S(M_1) \cap S(M_2)$ is as well,
and emptiness can be decided in these (different initial states can be used). If $\MFam$
has a decidable infiniteness problem, then it can be tested whether
$S(M_1,M_2)$ is infinite.
\qed \end{proof}

In Proposition \ref{loadedunloaded}, many machine models are listed that are known to have regular store  languages, and the regular languages are closed under intersection
with a decidable emptiness and infiniteness problems \cite{HU}. Then, combined with Proposition \ref{generalresult},
the following is immediate:
\begin{proposition}
Let $\MFam$ be any of the following models:
$$\NPDA,\NSA, r\mbox{-flip } \NPDA, \NRBQA, \NRBTA.$$
If $M_1,M_2 \in \MFam$, then $S(M_1,M_2)$,
can be accepted by deterministic finite automata
(and hence, by polynomial time, constant space deterministic Turing machines). 
Furthermore, $\MFam$ has a decidable common store configuration problem, and a decidable common store configuration infiniteness problem.
\end{proposition}

Further, in Section \ref{sec:store}, several machine models (listed
in Proposition \ref{NCMStoreLanguages2} below) are shown to have
all store languages in $\NCM$, and it is known that $\NCM$ is closed
under intersection and has a decidable emptiness and infiniteness problem \cite{Ibarra1978}. It is also known that for all $\NCM$ $M$, 
there is some constant $c$ such that every input $w$ of length $n$
can be accepted in at most $cn$ steps \cite{Baker1974}. From this, it follows that
for every $\NCM$ language, there is a nondeterministic logspace bounded
Turing machine, and hence a deterministic polynomial time Turing machine
to accept it.

Therefore, the following is immediate by Proposition \ref{generalresult}:
\begin{proposition}
\label{NCMStoreLanguages2}
Let $\MFam$ be any of the following models:
$$\NPCM(k), r\mbox{-flip } \NPCM(k), \NRBSCM(k), \NRBQCM(k), \NRBTCM(k).$$
If $M_1,M_2 \in \MFam$, then $S(M_1,M_2)
 \in \LFam(\NCM)$, and it can be accepted by a nondeterministic logspace, or
deterministic polynomial time Turing machine.
Furthermore,
$\MFam$ has a decidable common store configuration problem, and a decidable common store configuration infiniteness problem. 
\end{proposition}

Note that all decidability results here  depend heavily on store languages only having counters in them. Indeed, even the context-free languages are not closed under intersection, but because the store languages of $\NPCM$ are in $\LFam(\NCM)$, which is closed under intersection, these properties can be decided.

\section{Reachability Problems}
\label{sec:pairs}

An important topic extensively  studied in the verification community 
is the development of algorithms for reachability problems, i.e.,  problems
such as, given two configurations $c_1$ and $c_2$ of a system, is $c_2$ reachable from $c_1$? 
Here we study these questions for various models. 

Let $M$ be a machine, and define the following sets:
\begin{eqnarray*}
T(M) &=& \{ (c_1, c_2) \mid \mbox{configuration~} c_2 \mbox{~is
reachable from~} c_1\},\\
T_a(M) &=& \{ (c_1, c_2) \mid \mbox{configuration~} c_2 \mbox{~is
reachable from~} c_1 \mbox{~in an accepting computation}\},\\
T_a^R(M) &=& \{ (c_1, c_2^R) \mid \mbox{configuration~} c_2 \mbox{~is
reachable from~} c_1 \mbox{~in an accepting computation}\},\\
\PAIR(M) &=& \{ c_1 \# c_2 \mid \mbox{configuration~} c_2 \mbox{~is
reachable from~} c_1 \mbox{~in an accepting computation}\},\\
\PAIR^R(M) &=& \{ c_1 \# c_2^R \mid \mbox{configuration~} c_2 \mbox{~is
reachable from~} c_1 \mbox{~in an accepting computation}\}.
\end{eqnarray*}
All of these definitions refer to the same basic reachability problem.
However, as we shall show below, the ``difficulty'' (complexity and decidability) 
of accepting these sets depends on how 
the configurations are specified --- whether they are given as a 
tuple (i.e., on separate tapes) or given as a single string separated 
by a special marker $\#$,  and whether one of the configurations is given in reverse.  Also, are there any differences depending on whether the configurations are restricted to be only in accepting computations, or occurring in any computation?

Let $k \ge 1$.  A $k$-tape machine $M$ (of some type) is
a generalization of a 1-tape machine in that the machine now
has $k$ input tapes, each with a one-way read-only head. A
move of the machine now depends on the symbols scanned by
the $k$ heads. We assume, without loss of generality, that
at each step, at most one input head moves right on the tape.
The machine now accepts $k$-tuples of words.  Thus, an $\NFA$,
$\NCM$, $\NPCM$, etc.  which are 1-tape machines generalize
to $k$-tape $\NFA$, $k$-tape $\NCM$, $k$-tape $\NPCM$, etc.
Multi-tape machines are used to accept sets $T(M), T_a(M)$, and $T_a^R(M)$,
and single tape machines for $\PAIR(M)$ and $\PAIR^R(M)$.

First, consider the following result
from \cite{IbarraDangTCS}.
\begin{proposition} \cite{IbarraDangTCS}
If $M$ is an $\NPCM$, then the  set
$T(M) = \{ (c_1, c_2) ~|~$ configuration $c_2$
is reachable from $c_1\}$ can be accepted by a 2-tape $\NCM$.
\label{2tape}
\end{proposition} 
Next, it is seen that restricting this set to accepting computations does not change acceptance by $2$-tape $\NCM$.
\begin{proposition} If $M$ is an $\NPCM$, then $T_a(M)$
 can be accepted by a 2-tape $\NCM$.
\label{2tapeaccepting}
\end{proposition}
\begin{proof}
In the proof of Proposition \ref{StoresOfNPCM}, it was implied that
$\{c \mid c \mbox{~is a reachable configuration of } M\}$ can be accepted by an
$\NCM$ $M_1$ and $\{c \mid c \mbox{~can
reach a final configuration of } M\}$
can be accepted by an $\NCM$ $M_2$.
 
Now by Proposition \ref{2tape}, $T(M)$
can be accepted by a  2-tape $\NCM$ $A$. Then, 
a 2-tape $\NCM$ $A'$ can be constructed from $A$ which simulates $A$ but also simultaneously
simulates the $\NCM$ $M_1$ using additional counters to check
that the first configuration is reachable from the initial
configuration and simulates $M_2$
to check that the second configuration reaches an
 accepting configuration. This has the effect of restricting $A$ to exactly the accepting configurations.
\qed \end{proof}

Next, it will be
shown that there is an $\NPDA$ (respectively a $1$-flip
$\NPDA$) such that  $T_a^R(M)$ (respectively neither $T_a(M)$ nor $T_a^R(M)$) can be accepted by a 
2-tape $\NCM$, in contrast to Proposition \ref{2tapeaccepting}.
Three technical lemmas are first required.

\begin{lemma} \label{pal1}
Let $A$ be a 2-tape $\NCM$.  There exists a constant
$c$ such that any tuple $(x,y)$ accepted by $A$ can be 
accepted by $A$ within $cn$ time (i.e., number of steps),
where $n = |xy|$.
\end{lemma}
\begin{proof}
Assume, without loss of generality, that at each step, 
$A$ moves at most one tape head to the right. (The 
finite-state control has the specification of which tape
head reads.)  Construct an $\NCM$ $M$
(hence there is only one input head) which, when given
a string $w$, which is an interlacing of the symbols comprising
the input tuple $(x,y)$, simulates the movements of the
two heads of $A$ on $x$ and $y$ faithfully. Then $M$ accepts if and 
only if $A$ accepts.  Since $M$ is an $\NCM$, there is
a constant $c$ such that any input $w$ that is accepted by $M$ 
can be accepted within $cn$ steps, where $n = |w|$ \cite{Baker1974}.
Since $M$ faithfully simulates $A$, it follows that if
$(x,y)$ is accepted by $A$, then there an accepting computation of
$(x,y)$ that runs in $cn$ time, where $n = |xy|$.
\qed \end{proof}

Let $h \ge 1$.  An $h$-head machine $M$ (of some type) is
a generalization of a 1-head machine in that the machine now
has $h$ independent one-way read-only input heads
operating on a single input tape.  A
move of the machine now depends on the symbols scanned by
the $h$ heads. 
The machine accepts an input if all the heads falls off
the input in an accepting state. 
Thus, an $\NFA$,
$\NCM$, $\NPCM$, etc.\  can generalize from 1 to 
$h$-head $\NFA$, $h$-head $\NCM$, $h$-head $\NPCM$.

The following relates 2-tape  $\NCM$s to multi-head $\NCM$s.

\begin{lemma} \label{pal2}
If $T \subseteq \{(x, y) ~|~ x, y \in \Sigma^* \}$
is accepted by a 2-tape $\NCM$, then the language
$L = \{x \# y ~|~ (x, y) \in T \}$ (where $\#$ is a new
symbol) can be accepted by a multi-head $\NFA$.
\end{lemma}
\begin{proof}
Suppose $T$ can be accepted by a 2-tape $\NCM$ $A$.
Assume that $A$ has $k$ 1-reversal counters and when it
accepts, all counters are zero. Construct from $A$, a one-way 2-head $\NCM$ $M$ accepting
the language $L$. $M$ simply dispatches one head to 
the position of $\#$ and then simulates $A$ using its two
heads. From Lemma \ref{pal1}, there is a constant 
$c$ such that any tuple $(x,y)$ accepted by $A$ can be 
accepted within $cn$ time, where $n = |xy|$.  Hence the
value in each counter of $A$ during such an accepting
computation is bounded by $cn$ and, hence, each counter
of $M$ is bounded by $cn$.  Then we can replace 
each counter in $M$ by two (one-way) heads, where one 
head is used to simulate the counter increments and 
the other head is used for counter decrements (moving forward one cell
for every $c$ increased and decreased respectively).  Hence 
$L$ can be accepted by a ($2k+2$)-head $\NFA$.
\qed \end{proof}

The following known result is needed as well. The proof, 
which uses the ideas in \cite{YaoRivest}, was given
recently in \cite{Kutrib}.  The result can also
be shown using Kolmogorov complexity techniques \cite{MingLi}.
\begin{lemma} \cite{Kutrib} \label{pal3}
$L = \{x \# x^R ~|~ x \in \{0,1\}^+ \}$ cannot
be accepted by a multi-head $\NFA$. 
\end{lemma}
From Lemmas \ref{pal2} and \ref{pal3}, the following is immediate:
\begin{proposition} \label{pal4}
$T = \{(x, x^R) ~|~ x \in (0+1)^+ \}$
cannot be accepted by any 2-tape $\NCM$.
\end{proposition}

From Proposition \ref{2tapeaccepting}, it has already been shown
that if $M$ is an $\NPCM$, then 
$T_a(M) =  \{(c_1, c_2) ~|~c_2$ is
reachable from $c_1$ in some accepting computation$\}$
can be accepted by a 2-tape $\NCM$, 
However, the next result
provides a contrast.

\begin{proposition} \label{pal5}
There is a 1-flip $\NPDA$ $M$ such that
neither $T_a(M)$ nor $T_a^R(M)$ can be accepted
by a 2-tape $\NCM$.
In addition, there is an $\NPDA$ $M$ such that $T_a^R(M)$ cannot be accepted by a 2-tape $\NCM$.
\end{proposition}
\begin{proof}
Construct a 1-flip $\NPDA$ $M$ which, on $\lambda$
input, pushes $\#x\$$ on the stack for some $x = a_1 \cdots a_n$, for
nondeterministic $n \ge 1$ and nondeterministically selected $a_i$'s
from the alphabet $\{0,1\}$, and $\#,\$$ are special symbols, and then switches to some new state $q_0$. Then $M$ nondeterministically
does one of the following:
\begin{itemize}
\item enters an accepting state $f_1$
\item flips the stack and enters accepting state $f_2$.
\end{itemize}
$T_a(M)$ cannot be accepted by a 2-tape $\NCM$. Otherwise,
$$T_a(M) \cap \{(q_0\#y\$, f_2\$ z\# ) ~|~ y, z 
\in (0+1)^+\} = \{ (q_0\#x\$, f_2\$x^R\#) ~|~ x \in (0+1)^+\}$$
can also be accepted by a
2-tape $\NCM$, from which, another 2-tape $\NCM$ can be constructed
accepting $T = \{ (x, x^R) ~|~ x \in (0+1)^+\}$, which contradicts Proposition \ref{pal4}.

Similarly, $T_a^R(M)$ cannot be accepted by a 2-tape $\NCM$. Otherwise,
$T_a^R(M) \cap \{(q_0\#y\$, f_1\$ z\# ) ~|~ y, z 
\in (0+1)^+\} = \{ (q_0\#x\$, f_1\$x^R\#) ~|~ x \in (0+1)^+\}$
can also be accepted by a
2-tape $\NCM$.
Essentially this same proof works for $\NPDA$s as $M$ does not need to flip its pushdown for $T_a^R(M)$.
\qed \end{proof}

It is easy to show that the converse of Lemma \ref{pal2} is not true:

\begin{proposition} 
There is a language $L \subseteq (0+1+\$)^+ \# (0+1+\$)^+$
that is accepted by a 2-head $\DFA$ such that
$T = \{(x,y ) ~|~ x \# y \in L \}$ cannot be accepted by any
2-tape $\NCM$.
\end{proposition}
\begin{proof}
Let $L = \{w\$w \# z ~|~ w, z \in (0+1)^+ \}$.
Clearly, $L$ can be accepted by a 2-head $\DFA$.
Suppose $T = \{(w\$w,z) ~|~ w, z \in (0+1)^+ \}$
can be accepted by a 2-tape $\NCM$. Then a (1-tape) $\NCM$ can be constructed accepting the
language $L' = \{wcw ~|~ w \in (0+1)^+ \}$ by
just guessing the symbols comprising $z$ in a bit-by-bit fashion.
However, $L'$ cannot be accepted by any
$\NCM$ as has been seen in the proof of Proposition \ref{NPCMPAIR}.
\qed \end{proof}

Next, consider the $\PAIR(M)$ function.
Note that one could alternatively define a set similarly to $\PAIR(M)$,
say $\PAIR'(M)$, where it is just enforced that the second configuration can follow from the first, but not necessarily in an accepting computation.
However, in any class of one-way nondeterministic machines, 
$\PAIR'(M)$ is equal to $\PAIR(M')$ where $M'$ is obtained from $M$ by nondeterministically guessing some configuration and putting it in
the store, and letting all states be final. Therefore, for common classes
$\PAIR(M)$ is the more general definition.
Also note that membership
of $\PAIR^R(M)$ is decidable if and only if
membership of $\PAIR(M)$ is decidable
(if it is possible to decide whether $w \in \PAIR(M)$, then to decide if
$u \#v \in \PAIR^R(M)$, instead decide whether $u\#v^R \in \PAIR(M)$, which is equivalent).

It is now shown that for all $r$-flip $\NPCM$s $M$, $\PAIR^R(M)$ is in $r$-flip $\NPCM$. Here
$r=0$ gives the same result for $\NPCM$.

\begin{proposition} \label{NPCMPAIRR} Let $r \geq 0$.
If $M$ is an $r$-flip $\NPCM$, then $\PAIR^R(M)$ can be accepted by a $r$-flip $\NPCM$.
\end{proposition}
\begin{proof}
Let $M$ be an $r$-flip $\NPCM$ with $k$ counters. The store language
$S(M)$ of $M$ can be accepted by an $\NCM$ $M_1$
with $l_1$ counters (for some $l_1$) by Proposition \ref{flipstore}. Since $\NCM$ is closed under reversal
\cite{Ibarra1978}, $S(M)^R$ can also be accepted by an $\NCM$ $M_2$
with $l_2$ counters (for some $l_2$).
Construct another $r$-flip
$\NPCM$ $M_3$ with $k + l_1 + l_2$ counters, which on 
input $c_1 \# c_2^R$, pushes the pushdown contents
of $c_1$ into the pushdown and the counter contents into the first $k$ counters, and switches to the state in $c_1$.
While pushing, $M_3$ simultaneously simulates $M_1$ (using $l_1$ counters)
to check that $c_1$ is  in $S(M)$ (and thus reachable from the initial 
configuration of $M$).
Then, $M_3$ simulates $M$ until some nondeterministically guessed spot,
where the following are verified in parallel: 
\begin{itemize}
\item that the current configuration matches $c_2$ --- this is done by subtracting one from a counter for every one of that counter letters 
read from the input, then matching the pushdown contents in reverse --- and
then matching the state in $c_2$.
\item that $c_2^R$ is in $L(M_2)$ using the last $l_2$ counters.
\end{itemize}

The first statement verifies that the second configuration follows from the first. But to verify that the second configuration can lead to an accepting state,
it instead verifies that that configuration is in the store language.
\qed \end{proof}
Notice that without the proof that all $r$-flip $\NPCM$s have store languages in 
$\LFam(\NCM)$, it is not clear how the proof above could work. Indeed, after popping the current pushdown while
matching it to $c_2$, there is no way to continue the simulation 
(to an accepting configuration) starting
at $c_2$. However, since the store language only uses counters, that
can be checked with only extra counters in parallel.

Since the store language of every $\NPDA$ is regular, plus $\NPCM$ is a special case of
Proposition \ref{NPCMPAIRR} with $r=0$,
the following is immediate.
\begin{corollary}
For all $M \in \NPDA$, $\PAIR^R(M) \in \NPDA$. Also, for all $M \in \NPCM$, $\PAIR^R(M) \in \NPCM$.
\end{corollary}

However, it is not possible to accept $\PAIR(M)$ or $\PAIR^R(M)$ without the pushdown.
\begin{proposition} 
\label{NPCMPAIR}
There is an $\NPDA$ (hence, an $\NPCM$) $M$ such that both $\PAIR(M)$
and $\PAIR^R(M)$ are not in $\NCM$.
\end{proposition}
\begin{proof}
Consider an $\NPDA$ $M$ which, from the initial state $q_0$,
pushes $axa$ on the stack, where $x$ is a nondeterministically
guessed word in  $\{0,1\}^+$, and $M$ enters state $p$ for the first and only time.
Then $M$ pops $a$, pushes $b$,  and enters an accepting state $f$.

Suppose $\PAIR(M)$ can be accepted by an $\NCM$ $M'$. 
Let $L = \{paxa\#fyb ~|~ x, y \in  (0+1)^+ \}$, a regular language.
Another $\NCM$ $M''$ that
accepts $\PAIR(M) \cap L$ can be constructed as $\LFam(\NCM)$ is closed under
intersection with regular languages.  However, it is easy to show
that $L(M'') = \{ paxa\# faxb \mid x \in \{0,1\}^+ \}$
cannot be accepted by an $\NCM$ as follows. Certainly, if $L(M'')$   
is accepted by an $\NCM$, then $L = \{x \# x ~|~ x \in \{0,1\}^+\}$
can also be accepted by an $\NCM$ $M_L$.
Any string $v$ accepted by $M_L$ can be accepted within 
time linear in $|v|$ \cite{Baker1974}. Thus each
counter will have value at most linear in $|v|$.
So when $M_L$ is given string $v = x\#x$, 
where $x$ is of length $n$, the 
number of possible distinct configurations
when the input head of $M_L$ reaches $\#$
is at most  $cn^k$, where $k$ is the number of counters and $c$ is some constant.
Since the number of strings $w$ of length 
$n$ is $2^n > c n^k$ for all but a finite number of $n$'s, it follows that for some $x, x'$
with $|x| = |x'|$, $x \ne x'$, $M_l$ will
also accept $x\#x'$, which is a contradiction.

Similarly, for $\PAIR^R(M)$, $\{ paxa\# bx^Raf \mid x \in \{0,1\}^+ \}$
is not an $\NCM$ language.
\qed \end{proof}

In addition, given an $r$-flip pushdown machine $M$, 
it is possible to accept $\PAIR(M)$ with a
$r+1$-flip pushdown, as flipping the stack allows it to be matched
to the input configuration in order.
\begin{proposition}
If $M$ is an $r$-flip $\NPCM$, then $\PAIR(M)$ can
be accepted by a $r+1$-flip pushdown automaton. 
\end{proposition}

The following are all similar to the proofs of Propositions
\ref{NPCMPAIRR} and \ref{NPCMPAIR}, using corresponding results on store languages
of each type of machine to verify that the first configuration of the pair
is reachable, and the second can reach a final
configuration, while simulating $M$ between the pair of configurations. For stack automata, a machine accepting $\PAIR(M)$ or $\PAIR^R(M)$ needs to guess and mark the position of the read head in the second configuration of the pair when the symbol is getting pushed.
\begin{proposition}
Let $\MFam$ be any of the following machine models:
$$\NRBTA, \NSA, \NRBTCM, \NRBSCM.$$
If $M \in \MFam$, then $\PAIR(M)$ and $\PAIR^R(M)$ are in $\LFam(\MFam)$.
\end{proposition}
Similarly, for $\NRBQA$  (respectively $\NRBQCM$), then for every 
$M \in \NRBQA$, $\PAIR(M) \in \LFam(\NRBQA)$. However, we conjecture that
$\PAIR^R(M)$ need not be in $\LFam(\NRBQA)$. This is because, when dequeueing the
current configuration, it cannot match the reverse of the the input configuration.

Next, for one unrestricted counter plus reversal-bounded counters, then following
 even stronger result is obtained.
\begin{proposition}
Let $M \in \NCACM$. Then $\PAIR(M)$ and $\PAIR^R(M)$ are in $\DCM$. They
can also be accepted by a deterministic logspace bounded or polynomial time Turing machine.
\end{proposition}
\begin{proof}
It follows from the proof of Proposition \ref{NPCMPAIRR} that $\PAIR^R(M)$ is in 
$\NCACM$, and similarly to that proof, $\PAIR(M)$ is also in $\NCACM$, since the unrestricted counter only has one letter and can therefore
be read in reverse. However, it is known that all bounded $\NPCM$
languages are in $\DCM$ \cite{IbarraSeki}, and indeed both
$\PAIR(M)$ and $\PAIR^R(M)$ are bounded when all stores are counters.
\qed \end{proof}

\section{Reachability in Multi-Pushdown Machines}

In this section, one more general machine model is considered.
Let $n \ge 1$.  An  $n$-$\pd$ $M$ is a generalization of an $\NPDA$.  It has
a one-way input and $n$ pushdown  stacks  $P_1, \ldots,  P_n$.   The machine
starts with the first stack $P_1$ containing the start stack symbol, $Z_0$,
and the other stacks being empty, i.e., containing the string $\lambda$.
A move of $M$ consists of the following:
(i) reads a symbol or $\lambda$ from the input;
(ii) reads and pops the symbol on top of the first non-empty $P_i$.
(Thus, if $P_1$ is empty, it reads the top symbol of $P_2$
if it is non-empty, etc.);
(iii) changes state;
(iv) for each $1 \le i \le n$, writes (i.e., pushes) a finite-length string $\alpha_i$
on stack $P_i$. (Note that writing is a push move.) 
The  machine accepts if after reading
all symbols of the input, it eventually enters an accepting state.
We believe this model was first introduced and studied in \cite{multipushdown}
and has since been investigated in several places in the literature.
It is known that the emptiness and infiniteness problems for such machines are
decidable, and the languages accepted have an effectively computable
semilinear Parikh map \cite{multipushdown} (the proof of decidability for emptiness in \cite{multipushdown} contained an error which was corrected in \cite{AtigMultiPushdown}).  These results still hold when we have an $n$-$\pd$ augmented
with reversal-bounded counters; call this an $n$-$\pd$ $\NPCM$ (considered in \cite{Harju}).

In \cite{multipushdown}, these machines are referred to as ``ordered multi-pushdown
machines''. There, it is shown that for regular sets of configurations $C$,
$\pre_M^*(C)$ must also be regular. This is clearly not the case for $\post^*_M$, as
 a machine  $M$ could be built that pushes some nondeterministically guessed
string $a^n$ on the first pushdown in state $q_0$, then pop each $a$ while pushing $b,c,d$
to the second, third, and fourth pushdown respectively, and then switching to final state $f$ when the first pushdown is empty. Then
$\post_M^*(q_0 a^*) \cap f\{b,c,d\}^* = S(M) \cap f\{b,c,d\}^* = \{b^n c^n d^n \mid n \geq 0\}$, a non-regular language. Hence, $S(M)$ is not necessarily regular (or even
context-free), nor is $\post_M^*(C)$ when $C$ is regular.

But both of these types of machines can accept their own store
languages by increasing the number of pushdowns (and counters).
\begin{proposition} \label{storemulti}
If $M$ is an $n$-$\pd$ (resp.\ $n$-$\pd$ $\NPCM$ with $k$ counters),
then $S(M)$ can be accepted by a $3n$-$\pd$ (resp.\ $3n$-$\pd$ $\NPCM$ with $2k$ counters).
\end{proposition}
\begin{proof}
Given such a machine $M$ with $n$ pushdowns and $k$ counters, 
construct $M'$ to accept $S(M)$ with $3n$ pushdowns and $2k$ counters. On 
$\lambda$ transitions, $M'$ starts by simulating $M$
using the first $n$ pushdowns and $k$ counters. Then at some arbitrary spot, $M'$
moves each pushdown, one at a time, to pushdown $n+1$ to $2n$, so that
each pushdown contents becomes reversed. Then, it matches the pushdowns to the 
pushdown part of the input configuration (which is now in the correct order), while
in parallel, again reversing the pushdown contents using the final $n$ pushdowns. Then,
it matches the first $k$ counters to the input, while in parallel making a copy of
each of the $k$ counters using counters $k+1$ to $2k$. Lastly, $M'$ is able
to continue the simulation of $M$ using the last $n$ pushdowns and $k$ counters.
\qed \end{proof}

However, next we see that the family has an undecidable common reachability
problem with two pushdowns and no counters.
\begin{proposition}
The common store configuration problem is undecidable
for $2$-$\pd$. Similarly for the common store configuration infiniteness problem.
\end{proposition}
\begin{proof}
It is known that it is undecidable whether the intersection of two $\NPDA$s is
empty \cite{HU}. Then given two $\NPDA$s $M_1$ and $M_2$, where, without loss of generality,
both have the same unique final state $f$, and all other states in $M_1$ are not used in $M_2$, and vice versa, and that the pushdowns of both machines empty before switching
to $f$. Then, construct two $2$-$\pd$ machines $M_1'$ and $M_2'$ that simulate
$M_1$ and $M_2$ respectively, while copying the input to the second pushdown as
it reads it. Any common reachable configurations must be in state $f$ with the
first pushdown empty and the input on the second pushdown. Then 
$S(M_1') \cap S(M_2') = \emptyset$ if and only if $L(M_1) \cap L(M_2) = \emptyset$,
which is undecidable. For the common reachability infiniteness problem, introduce
a new pushdown symbol $c$ to create $M_1''$ from $M_1'$, and $M_2''$ from $M_2'$, and have these new machines push arbitrarily many $c$'s onto the second pushdown at the end of the computation. Then $S(M_1') \cap S(M_2') = \emptyset$ if and only if
$S(M_1'') \cap S(M_2'') = \emptyset$, but also, if the latter is non-empty, then it
must be infinite. 
\qed \end{proof}
Thus, despite the store languages having a decidable emptiness
problem, it is not possible to determine if there are any
common reachable configurations between
two machines.

Lastly, we see that it is possible to accept $\PAIR(M)$ by increasing the number of pushdowns and counters. This proof is similar to Proposition \ref{storemulti}.
\begin{proposition}
If $M$ is an $n$-$\pd$ (resp.\ an $n$-$\pd$ $\NPCM$ with $k$ counters),
then $\PAIR(M)$ can be accepted by a $5n$-$\pd$ (resp.\ a $5n$-$\pd$ $\NPCM$ with $3k$ counters).
\end{proposition}
\begin{proof}
This proof is similar to that of Proposition \ref{storemulti}.
Given $M$, create $M'$ that on input $c_1 \# c_2$, first simulates $M$ on $\lambda$ transitions using the
first set of pushdowns and counters. Then, at an arbitrary spot, $M'$ reverses
the contents of each pushdown, then matches those to $c_1$ while in parallel
reversing their contents using the third set of pushdowns. Then, it makes
a copy of the contents of the counters using the second set of counters while matching their contents to the counters of $c_1$. $M'$ continues the simulation using the third set of pushdowns and second set of counters until another arbitrary spot, where $M'$ repeats the same procedure using the fourth, then fifth set of pushdowns, and the third set of counters, matching to $c_2$, and continuing the simulation.
\qed \end{proof}

\section{Conclusion}

Store languages of a new machine model combining an $\NPDA$ that can flip its contents a bounded number of times together with reversal-bounded counters are investigated.
The store languages can be accepted by machines with only 
reversal-bounded counters (and no pushdown).
In addition, general connections were established between the notion of the store language
of a machine model, and reachability/verification problems in infinite-state systems.
In particular, store languages were connected to the problem of accepting the configurations that can be reached from (or can reach) a given regular set of configurations. The connection allows for several more general results than what is known in the literature. For example, the successor and predecessor configurations of a stack automaton from a given regular set of configurations must be a regular language.
Several models augmented
by counters were also shown to accept successor and predecessor configurations by eliminating the main store, similarly leading to
decidable reachability properties.

Some interesting open problems remain. In particular, the time and space complexity of
constructing store languages from a given type of machine has not yet been investigated.

\section*{Acknowledgements} We thank the anonymous reviewers for their suggestions which improved the presentation of the paper.

\bibliography{undecide_refs}{}
\bibliographystyle{elsarticle-num}

\end{document}